\documentclass[11pt]{article}
\usepackage[utf8]{inputenc}
\usepackage[margin=1.in]{geometry}
\usepackage{amsmath,amsthm,amsfonts,amssymb,amscd}
\usepackage{lastpage}
\usepackage{enumerate}
\usepackage{fancyhdr}
\usepackage{mathrsfs}
\usepackage{xcolor}
\usepackage{graphicx}
\usepackage{listings}
\usepackage{hyperref}
\usepackage{tikz}
\usepackage[style=alphabetic, maxbibnames=5]{biblatex}
\usepackage[linesnumbered,ruled,vlined]{algorithm2e}
\SetKwInput{KwInput}{Input}                
\SetKwInput{KwOutput}{Output}
\SetKwFor{RepTimes}{repeat}{times}{end}

\usepackage{amssymb,amsmath,amsthm, bbm}
\usepackage{color, xcolor}
\usepackage{graphicx}
\usepackage{setspace}
\usepackage{xspace}
\usepackage{multirow}
\usepackage{array}
\usepackage{complexity}
\usepackage{cleveref}
\usepackage{dirtytalk}

\theoremstyle{plain}
\newtheorem{theorem}{Theorem}[section]
\newtheorem{corollary}[theorem]{Corollary}

\newtheorem{lemma}[theorem]{Lemma}
\newtheorem{claim}[theorem]{Claim}
\newtheorem{fact}[theorem]{Fact}

\newtheorem{definition}[theorem]{Definition}

\theoremstyle{remark}
\newtheorem{remark}[theorem]{Remark}
\theoremstyle{plain}
\newclass{\DNF}{DNF}
\newclass{\DNFs}{DNFs}
\newclass{\ACzero}{AC^0}
\newclass{\TCzero}{TC^0}

\renewcommand{\R}{\mathbb{R}} 
\newcommand{\N}{\mathbb{N}} 

\renewcommand{\Pr}{\mathop{\bf Pr\/}}
\renewcommand{\E}{\mathop{\bf E\/}}
\def\var{{\mathop{\bf Var\/}}}

\newcommand{\Var}{\mathop{\bf Var\/}}

\newcommand{\bigbracket}[1]{\Bigl [ #1 \Bigr ]}
\newcommand{\bigparen}[1]{\Bigl ( #1 \Bigr )}

\newfunc{\MAJ}{MAJ}
\newfunc{\MUX}{MUX}
\newfunc{\NAE}{NAE}
\newfunc{\OR}{OR} 
\newfunc{\AND}{AND}
\newfunc{\Tribes}{Tribes}
\newfunc{\LocalCorrect}{LocalCorrect}

\newfunc{\sgn}{sgn} 

\newfunc{\spar}{sparsity}
\newfunc{\rank}{rank}
\newfunc{\spn}{span}
\newfunc{\quasipoly}{quasipoly}
\newfunc{\Bias}{Bias}

\newfunc{\DT}{DTdepth} 
\newfunc{\DTs}{DTsize} 
\renewcommand{\W}{\mathbf{W}} 
\newcommand\Inf{{\mathbf{Inf}}} 
\newcommand\NormInf{{\mathbf{NInf}}} 

\newcommand{\eps}{\varepsilon}

\renewcommand{\hat}{\widehat}
\renewcommand{\tilde}{\widetilde}

\newcommand{\pmone}{\{\pm 1\}}

\newcommand{\calA}{\mathcal{A}}

\newcommand{\calC}{\mathcal{C}}
\newcommand{\calD}{\mathcal{D}}
\newcommand{\calE}{\mathcal{E}}

\newcommand{\calJ}{\mathcal{J}}

\newcommand{\calL}{\mathcal{L}}

\newcommand{\calR}{\mathcal{R}}
\newcommand{\calS}{\mathcal{S}}

\newcommand{\bfx}{\mathbf{x}}
\newcommand{\bfy}{\mathbf{y}}

\newfunc{\Parity}{PARITY}

\newfunc{\Dict}{Dict}
\newfunc{\Corr}{Corr}
\newfunc{\avg}{avg}
\newfunc{\smooth}{smooth}
\newfunc{\dist}{dist}

\newcommand{\argmax}{\textrm{argmax}}
\newcommand{\InfluentialCoords}{\calS}
\newcommand{\InfluentialOracles}{\mathcal{D}}
\newcommand{\Juntas}[2]{\mathcal{J}_{#1, #2}}
\newcommand{\Restrictedf}{f_{\bar{J}\to z}}

\newcommand{\BooleanHypercube}[1]{\{\pm 1\}^{#1}}
\newcommand{\specialset}{T}
\newcommand{\Binomial}{\textrm{Bin}}

\newcommand{\corr}{\mathsf{corr}}
\newcommand{\low}{\mathsf{low}}

\newclass{\ETH}{ETH}

\newcommand{\fourierCutoff}{\kappa}
\newcommand{\branchProcessDepth}{\alpha}

\renewcommand{\epsilon}{\varepsilon}

\allowdisplaybreaks
\title{Junta Distance Approximation with Sub-Exponential Queries}
\author{
Vishnu Iyer\thanks{UC Berkeley. Email: \texttt{vishnu.iyer@berkeley.edu}}
\and
Avishay Tal\thanks{UC Berkeley. Email: \texttt{atal@berkeley.edu}}
\and
Michael Whitmeyer\thanks{UC Berkeley. Email: \texttt{mwhitmeyer@berkeley.edu}.}
}

\hypersetup{%
  colorlinks=false,
  linkcolor=red,
  linkbordercolor={0 0 1}
}

\begin{document}


\maketitle

\begin{abstract}
 Leveraging tools of De, Mossel, and Neeman [FOCS, 2019], we show two different results pertaining to the \emph{tolerant testing} of juntas. Given black-box access to a Boolean function $f:\{\pm1\}^{n} \to \{\pm1\}$: \begin{enumerate}
     \item We give a $\poly(k, \frac{1}{\varepsilon})$ query algorithm that distinguishes between functions that are $\gamma$-close to $k$-juntas and $(\gamma+\varepsilon)$-far from $k'$-juntas, where $k' = O(\frac{k}{\varepsilon^2})$.
     \item In the non-relaxed setting, we extend our ideas to give a  $2^{\tilde{O}(\sqrt{k/\varepsilon})}$ (adaptive) query algorithm that distinguishes between functions that are $\gamma$-close to $k$-juntas and $(\gamma+\varepsilon)$-far from $k$-juntas. To the best of our knowledge, this is the first subexponential-in-$k$ query algorithm for approximating the distance of $f$ to being a $k$-junta (previous results of Blais, Canonne, Eden,  Levi,  and  Ron [SODA, 2018] and De, Mossel, and Neeman [FOCS, 2019] required exponentially many queries in $k$).
 \end{enumerate}  
 Our techniques are Fourier analytical and make use of the notion of ``normalized influences'' that was introduced by Talagrand \cite{Talagrand-zero-one}. 
 
\end{abstract}


\section{Introduction}
\label{section:introduction}


The study of property testing, initiated by Blum, Luby, and Rubinfeld in their seminal work on linearity testing \cite{DBLP:conf/stoc/BlumLR90}, is concerned with making fast decisions about a global object having some global property, while only accessing (or \say{querying}) parts of it. This notion was further explored by Goldreich, Goldwasser, and Ron~\cite{DBLP:conf/focs/GoldreichGR96}, who drew connections to the areas of learning theory and approximation algorithms in the context of graph properties.
We focus on properties of Boolean functions, i.e.,  $f:\BooleanHypercube{n} \to \BooleanHypercube{}$.
First, we state the definition of a property testing algorithm $\calA$. Given $\epsilon > 0$ and a class of functions $\calC$, we say that $\calA$ is a property tester for $\calC$ if it satisfies the following two conditions: 
\begin{enumerate}
    \item if $f \in \calC$, then $\calA$ accepts $f$ with probability at least $2/3$;
    \item if $\dist(f,g) \geq \epsilon$ for all $g \in \calC$, then $\calA$ rejects with probability at least $2/3$. 
\end{enumerate}
In the above definition, $\dist(f,g) = \Pr[f(x) \neq g(x)]$ is the fraction of inputs on which $f$ and $g$ disagree under the uniform distribution. The primary measure of efficiency for such property testing algorithms is the algorithms \textit{query complexity}, or the number of times it must use its black box access to $f$. Such query algorithms can be \textit{adaptive} in that the coordinates on which they query $f$ depend on previous answers, or they can be \textit{nonadaptive} in that the algorithm always queries $f$ in a predetermined manner. 

In this writeup, our algorithms will be adaptive, and we will focus on testing the particular class of functions known as $k$-juntas. Juntas comprise a simple and natural class of functions: those that depend only on a smaller subset of their input variables. More precisely, a  Boolean function $f: \BooleanHypercube{n} \to \BooleanHypercube{}$ is said to be a $k$-junta if there exists $k$ coordinates $i_1,\ldots,i_k\in [n]$ such that $f(x)$ only depends on $x_{i_1},\ldots,x_{i_k}$. In essence, juntas capture the existence of many irrelevant variables, and arise naturally in the context of feature selection in machine learning and many computational biology problems. A canonical example is the problem of determining the relationship between genes and phenotypes; for example, one might wish to test whether a particular physical trait is a function of many genes or only a small number. 

The fundamental problem of learning and/or testing juntas has been given much attention in recent years. We refer the reader to the works of Mossel, O'Donnell, and Servedio \cite{learning-juntas-MOS} and Valiant \cite{Valiant-learning-juntas} for the most recent work on learning $k$-juntas. 
In this paper, we focus on the problem of testing juntas. Testing $1$-juntas (aka dictators) and related functions had initial theoretical interest in the context of long-code testing in PCPs \cite{hastad-inapproximability, BellareGS98}, and was first formally explored in \cite{ParnasRS01}, which gave algorithms for testing dictators, monomials, and monotone DNFs. 
The more general problem of testing $k$-juntas was first studied by Fischer et. al.~\cite{FischerKRSS04}, where they exhibited a $k$-junta tester with query complexity $\tilde{O}(k^2)$ queries to $f$. Crucially, their upper bound lacked any dependence on the ambient dimension $n$. 
More recently, it was shown in \cite{og-junta-testing} that $O(k \log k + k/\epsilon)$ adaptive queries suffice to test $k$-juntas, and this is tight for constant $\epsilon$ \cite{saglam18-juntatestlb, DBLP:journals/ipl/ChocklerG04}. 
There has also been recent interest in the distribution free setting for junta testing (wherein the distribution on inputs is not assumed to be uniform). 
Liu et al. \cite{DBLP:journals/talg/LiuCSSX19-dist-free} initially gave a $\tilde{O}(k^2/\eps)$-query algorithm with one-sided error, which was quickly followed up by the works of Bshouty \cite{DBLP:conf/coco/Bshouty19} and Zhang \cite{DBLP:zhang-dist-free-mab-2019} who gave $\tilde{O}(k/\eps)$-query algorithms with two-sided and one-sided error, respectively. 
The methods utilized by Bshouty extend those of Diakonikolas et al. \cite{DBLP:conf/focs/DiakonikolasLMORSW07} and result in algorithms not only for junta testing but also several subclasses of juntas. We note that while we solve a similar problem in a different setting, some of our techniques resemble those of \cite{DBLP:conf/coco/Bshouty19}: notably, an idea introduced in \cite{DBLP:conf/coco/Bshouty19} is to find a witness such that, if all coordinates outside a subset of the coordinates are fixed to this witness' values, then $f$ becomes a dictator on a single coordinate within that subset. This can be thought of as obtaining oracle access to a relevant coordinate, an idea pervasive throughout the work of \cite{junta-coordinate-oracles} and ours. The techniques in \cite{DBLP:conf/focs/DiakonikolasLMORSW07, DBLP:conf/icalp/DiakonikolasLMSW08, DBLP:conf/coco/Bshouty19} can all be categorized in the ``testing via implicit learning'' paradigm, as surveyed in \cite{DBLP:conf/propertytesting/Servedio10}.  

\subsection{Tolerant Junta Testing}

One of the first relaxations of the standard property testing model considered (sometimes referred to as the ``parameterized'' regime) were testers that distinguished between $f \in H$ and $f$ being $\eps$-far from $H'\supseteq H$. This notion was introduced by Kearns and Ron~\cite{DBLP:journals/jcss/KearnsR00} in the context of testing decision trees and certain classes of neural networks. We note that if $H'$ is a strict superset of $H$, then the job of the tester becomes easier, and smaller query or sample complexity is often achievable than in the regular testing model. Indeed, our \Cref{theorem:improved-dmn-intro} is an example of a (tolerant) parameterized tester.
\textit{Tolerant testing} is another generalization of the standard property testing model. The notion was first introduced by Parnas, Ron, and Rubinfeld~\cite{og-tolerant}. 
Normal property testing entails distinguishing between functions that \textit{exactly} satisfy a certain property, and functions that are $\epsilon$-far from satisfying said property. 
This is somewhat restrictive, and the tolerant testing problem seeks to more generally distinguish functions that are $c_\ell$ close to having the desired property, and those that are at least $c_u$ far from having the property, for some $0 < c_\ell < c_u < 1$.  
We also note that the notion of tolerant testing is closely related to the notion of distance approximation -- indeed, if one can estimate $\dist(f, \calC)$ up to additive error $(c_u - c_\ell)/2$ with probability at least $2/3$, then one has solved the tolerant testing problem for that class.\footnote{The reverse direction is also true -- given a tolerant tester it is possible to estimate the distance to that property. See for example section 3 in \cite{ailon-chazelle-monotonicity-tolerant}.} In general, tolerant testing (and therefore distance approximation), is much more challenging than traditional property testing. \Cref{figure:tolerant-testing-diagram} provides a visualization of the tolerant testing problem. 
Tolerant testing has received a lot of attention recently, see for example \cite{LYT-testing-reconstruction-DTs} for work on tolerant testing of decision trees and \cite{ailon-chazelle-monotonicity-tolerant} \cite{junta-lbs-PRW20} for work on tolerant testing of monotonicity. For the case of $k$-juntas, we have the following (relaxed) definition of a tolerant tester. In the following we denote by $\Juntas{n}{k}$ the class of $k$-juntas, and for a class of functions $\calC$, we denote $\dist(f,\calC) := \min_{g\in \calC}\dist(f,g)$.
\begin{definition}
\label{definition:tolerant-junta-testing}

For constants $0 < c_\ell < c_u < 1/2$ and a given $k',k\in \N$ with $k'\geq k$, a \emph{$(k, k', c_\ell, c_u)$ tolerant junta tester} is an algorithm that, given oracle access to $f: \BooleanHypercube{n} \to \BooleanHypercube{}$, 
\begin{enumerate}
    \item if $\dist(f, \Juntas{n}{k}) \leq c_\ell$ accepts with probability $2/3$;
    \item if $\dist(f, \Juntas{n}{k'}) \geq c_u$ rejects with probability $2/3$.
\end{enumerate}
\end{definition}
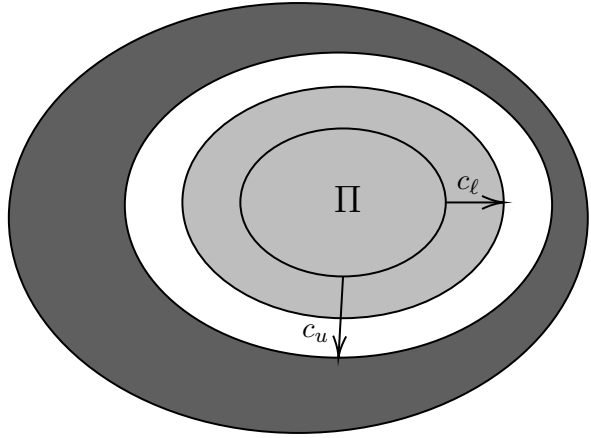
\begin{figure}
    \centering
    \tikzset{every picture/.style={line width=0.75pt}} 

\begin{tikzpicture}[x=0.75pt,y=0.75pt,yscale=-1,xscale=1]

\draw  [fill={rgb, 255:red, 95; green, 95; blue, 95 }  ,fill opacity=1 ] (159.5,150.5) .. controls (159.5,90.58) and (224.87,42) .. (305.5,42) .. controls (386.13,42) and (451.5,90.58) .. (451.5,150.5) .. controls (451.5,210.42) and (386.13,259) .. (305.5,259) .. controls (224.87,259) and (159.5,210.42) .. (159.5,150.5) -- cycle ;
\draw  [fill={rgb, 255:red, 255; green, 255; blue, 255 }  ,fill opacity=1 ] (218,144) .. controls (218,101.47) and (266.24,67) .. (325.75,67) .. controls (385.26,67) and (433.5,101.47) .. (433.5,144) .. controls (433.5,186.53) and (385.26,221) .. (325.75,221) .. controls (266.24,221) and (218,186.53) .. (218,144) -- cycle ;
\draw  [fill={rgb, 255:red, 190; green, 190; blue, 190 }  ,fill opacity=1 ] (247.04,142.66) .. controls (247.04,110.39) and (283.32,84.24) .. (328.07,84.24) .. controls (372.83,84.24) and (409.1,110.39) .. (409.1,142.66) .. controls (409.1,174.92) and (372.83,201.08) .. (328.07,201.08) .. controls (283.32,201.08) and (247.04,174.92) .. (247.04,142.66) -- cycle ;
\draw  [fill={rgb, 255:red, 190; green, 190; blue, 190 }  ,fill opacity=1 ] (276.27,142.66) .. controls (276.27,122.03) and (299.46,105.31) .. (328.07,105.31) .. controls (356.69,105.31) and (379.88,122.03) .. (379.88,142.66) .. controls (379.88,163.29) and (356.69,180.01) .. (328.07,180.01) .. controls (299.46,180.01) and (276.27,163.29) .. (276.27,142.66) -- cycle ;
\draw    (379.88,142.66) -- (407.1,142.66) ;
\draw [shift={(409.1,142.66)}, rotate = 180] [color={rgb, 255:red, 0; green, 0; blue, 0 }  ][line width=0.75]    (10.93,-3.29) .. controls (6.95,-1.4) and (3.31,-0.3) .. (0,0) .. controls (3.31,0.3) and (6.95,1.4) .. (10.93,3.29)   ;
\draw    (328.07,180.01) -- (325.86,219) ;
\draw [shift={(325.75,221)}, rotate = 273.24] [color={rgb, 255:red, 0; green, 0; blue, 0 }  ][line width=0.75]    (10.93,-3.29) .. controls (6.95,-1.4) and (3.31,-0.3) .. (0,0) .. controls (3.31,0.3) and (6.95,1.4) .. (10.93,3.29)   ;

\draw (322,133) node [anchor=north west][inner sep=0.75pt]  [font=\Large]  {$\Pi $};
\draw (306,203) node [anchor=north west][inner sep=0.75pt]    {$c_{u}$};
\draw (384,128) node [anchor=north west][inner sep=0.75pt]    {$c_{\ell }$};

\end{tikzpicture}
    \caption{A visualization of the tolerant property testing paradigm. Assuming the outermost oval represents all functions $f: \BooleanHypercube{n} \to \BooleanHypercube{}$ and the property at hand is represented by a class of functions $\Pi$, the goal is to distinguish between the light grey (at most $c_\ell$ close to a function in $\Pi$) and the dark grey (at least $c_u$ far from all functions in $\Pi$) regions.}
    \label{figure:tolerant-testing-diagram}
\end{figure}
Our definition incorporates both tolerant and parameterized testers; when $c_\ell = 0$ the tester is non-tolerant and when $k' = k$ the tester is non-parameterized. We note that in the above definition we upper bound $c_u<1/2$ since $k$-juntas are closed under complements, meaning if $g \in \Juntas{n}{k}$, then $-g \in \Juntas{n}{k}$. Parnas, Ron, and Rubinfeld in their seminal work \cite{og-tolerant}  showed that while standard property testers, when querying uniformly, are weakly tolerant, entirely new algorithms are usually needed to tolerant test with better parameters. Tolerant junta testing was first considered by Diakonikolas et al. \cite{DBLP:conf/focs/DiakonikolasLMORSW07} which used the aforementioned observation from \cite{og-tolerant} to show that a standard tester from \cite{FischerKRSS04} actually gave a $(k, k, \poly(\frac{\gamma}{k}), \gamma)$ tolerant tester. Chakraborty et al. \cite{DBLP:conf/coco/ChakrabortyFGM12} subsequently showed that a similar analysis to that of Blais \cite{og-junta-testing} gave a $(k, k, \gamma/C, \gamma)$ tolerant junta tester (for some constant $C$) using $\exp(k/\gamma)$ queries.

More recently, Blais et al. \cite[Theorem 1.2]{junta-submodular-opt} showed a tradeoff between query complexity and the amount of tolerance. In particular, they gave an algorithm which, given $k$, $\gamma$, and $\rho\in(0,1)$, is a $(k, k, \rho\gamma/16, \gamma)$ tolerant junta tester. The query complexity of the algorithm is $O\!\left(\frac{k\log k}{\gamma \rho (1-\rho)^k}\right )$. In particular, note that when $\rho$ is a constant bounded away from zero, this yields an $\exp(k)$ query algorithm, but when $\rho = 1/k$ this yields a $\poly(k)$ query algorithm. We also note that there is  an undesirable multiplicative \say{gap} between $c_u$ and $c_\ell$ that precludes one from tolerantly testing for arbitrary close values of $c_u$ and $c_\ell$ (i.e., in~\cite{junta-submodular-opt},  $c_u \ge 16 c_{\ell}$ for all choices of $\rho$). 
The recent work of \cite{junta-coordinate-oracles} addressed this, giving an algorithm for any arbitrary $\gamma, \eps>0$ that required $2^k\poly(k,\frac{1}{\eps})$ queries and was a $(k, k, \gamma, \gamma+\eps)$ tolerant junta tester.

In the relaxed setting (when $k' \neq k$), \cite[Theorem 1.1]{junta-submodular-opt} also gave an algorithm which used $\poly(k, \frac{1}{\gamma})$ queries to $f$ and was a $(k, 4k, \gamma/16, \gamma)$ tolerant junta tester. This once again posed the issue of not allowing for arbitrary $c_u$ and $c_\ell$ values, which was resolved by \cite[Corollary 1.6]{junta-coordinate-oracles}, which gave a $(k, O(k^2/\eps^2), \gamma, \gamma+\eps)$ tolerant junta tester with query complexity $\poly(k, \frac{1}{\eps})$. 

It is interesting to note that the techniques used to obtain the results from \cite{junta-submodular-opt} and \cite{junta-coordinate-oracles} are actually quite different, and yield results that are qualitatively similar but quantitatively incomparable. The results from \cite{junta-submodular-opt} extend the techniques of \cite{og-junta-testing}, which partition the $n$ input coordinates into $\poly(k)$ disjoint sets or \say{parts}. It is immediate that any $k$-junta is a $k$-part junta, but in \cite{og-junta-testing} it was shown that with high probability a function that is far from being a $k$-junta is also far from being a \say{$k$-part junta} (for a definition of this and more details we refer the reader to \cite{og-junta-testing}). 
The results of \cite{junta-submodular-opt} extend the idea of considering the relationship between $k$-juntas and $k$-part juntas in the context of tolerant testing. 

The techniques in \cite{junta-coordinate-oracles} suggest a new way of attacking the problem of tolerant $k$-junta testing. The core idea in \cite{junta-coordinate-oracles} was to get access to \say{oracles} to coordinates of $f$ which have large low-degree influence. These \emph{coordinate oracles} are obtained with high probability via a combination of random restrictions and noise operators to the original function, and once obtained, can be used to search, in a brute force manner, for the nearest $k$-junta. 

In terms of lower bounds for tolerant testing of juntas, two recent works addressed the non-adaptive case. Levi and Waingarten \cite{first-nonadaptive-junta-lbs} demonstrated that there exists $0<\eps_1<\eps_2<1/2$ such that any $(k, k, \eps_1,\eps_2)$ tolerant junta tester requires $\tilde{\Omega}(k^2)$ non-adaptive queries to $f$. In particular, this result demonstrated that the tolerant testing regime is quantitatively harder than the standard testing regime, in which a $\tilde{O}(k^{3/2})$-query non-adaptive query algorithm is known~\cite{ogg-testing-juntas} (and indeed optimal due to \cite{DBLP:journals/jacm/ChenSTWX18}). Subsequently, Pallavoor, Raskhodnikova, and Waingarten \cite{junta-lbs-PRW20} demonstrated that for any $k \le n/2$ there exists $0<\eps_1<\eps_2<1/2$ (with $\eps_1 = O(1/k^{1-\eta})$ and $\eps_2 = \Omega(1/\sqrt{k})$) such that every nonadaptive $(k, k, \eps_1,\eps_2)$-tolerant junta tester  requires at least $2^{k^\eta}$ queries to $f$, for any $0<\eta<1/2$.%
\footnote{
We note that this lower bound does not necessarily rule out $\poly(k)\exp(1/\eps)$ nonadaptive query $(k, k, \eps_1, \eps_2)$ (where $\eps = \eps_2-\eps_1$) tolerant junta testers due to the setting of $\eps_1$ and $\eps_2$ in their hard instance.}

\subsection{Our Results}

Our first result is a subexponential-in-$k$ query tolerant junta tester in the standard (non-relaxed) setting. In fact, we obtain an $\eps$-accurate estimate of the distance of $f$ to the class of $k$-juntas. 

\begin{theorem}
\label{theorem:main-result}
Given a Boolean function $f: \BooleanHypercube{n} \to \BooleanHypercube{}$, it is possible to estimate the distance of $f$ from the class of $k$-juntas to within additive error $\epsilon$ with probability $2/3$ using 
$ 2^{\tilde{O}(\sqrt{k/\varepsilon})}$ adaptive queries to $f$. In particular, when $\epsilon$ is constant, this yields a $2^{\tilde{O}(\sqrt{k})}$-query algorithm. However, the algorithm still requires $\exp(k/\eps)$ time.
\end{theorem}

A simple corollary of the above theorem is that for any $0 < c_\ell < c_u < 1/2$, we have a $(c_u, c_\ell, k,k)$ tolerant junta tester with the same query complexity as in \Cref{theorem:main-result}, where $\epsilon = (c_u - c_\ell)/2$. This is an improvement of the results of \cite{junta-coordinate-oracles, junta-submodular-opt}, whose tolerant junta testers when $k'=k$ required exponential query complexity in $k$ in the worst case. We note that although we obtain this improvement, our algorithm still requires $\exp(k)$ time. 
In the appendix, we show a result solving a similar problem\footnote{In particular, this problem is the problem of finding the subset of $k$ inputs that ``contain'' the most Fourier mass -- see \Cref{section:preliminaries} and \Cref{theorem:main-result-mass} for more details.} with an improved dependence on $\eps$, giving an algorithm requiring only $2^{\tilde{O}(\sqrt{k}\log(1/\eps))}$-queries and $\exp(k\log(1/\eps))$ time (see \Cref{theorem:main-result-mass}).

In the relaxed/parameterized setting when $k' \neq k$, we give a polynomial-in-$k$ query tolerant junta tester that is valid for any setting of $c_u$ and $c_\ell$, and reduces $k'$ dependence on $k$ to be linear instead of quadratic due to the result of \cite[Corollary 1.6]{junta-coordinate-oracles}.

\begin{theorem}
\label{theorem:improved-dmn-intro}
For any $\gamma,\eps>0$ and $k \in \N$, there is an algorithm with query complexity $\poly(k, 1 / \eps)$ that is a $(k , O(k/\eps^2), \gamma, \gamma+\eps)$-tolerant junta tester.
\end{theorem}

Theorem~\ref{theorem:improved-dmn-intro} is a simple corollary of the following theorem we prove.

\begin{theorem}\label{theorem:improved-dmn}
Let $\epsilon>0$, $k \in \N$, and $k' = O(k/\eps^2)$. Then, there exists an algorithm that given parameters $k, \eps$ and oracle access to $f$ makes at most $poly(k, 1 / \eps)$ queries to $f$ and returns a number $\alpha$ such that with high probability (at least $0.99$) 
\begin{enumerate}
	\item $\alpha \le \dist(f,\mathcal{J}_{n,k})+\epsilon$
	\item $\alpha \ge \dist(f,\mathcal{J}_{n,k'})-\epsilon$
\end{enumerate}
\end{theorem}
Indeed, to solve the problem in \Cref{theorem:improved-dmn-intro} we can apply the algorithm from \Cref{theorem:improved-dmn} with $\eps = (c_{u}-c_{\ell})/3$ and accept if and only if $\alpha < \frac{1}{2} (c_{u} + c_{\ell})$. If $\dist(f,\mathcal{J}_{n,k}) \le c_{\ell}$ we have that with high probability $\alpha \le c_{\ell} + \eps < \frac{1}{2} (c_{u} + c_{\ell})$ and we will accept.
On the other hand, 
if $\dist(f,\mathcal{J}_{n,k'}) \ge c_{u}$ we have that with high probability $\alpha \ge c_{u} - \eps > \frac{1}{2} (c_{u} + c_{\ell})$ and we will reject.

Both of the algorithms used to prove \Cref{theorem:main-result} and \Cref{theorem:improved-dmn} rely on the fact that we can get approximate oracle access to influential coordinates of $f$ using techniques from \cite{junta-coordinate-oracles}. From there, we analyze the Fourier coefficients of $f$ after a series of random restrictions in order gain more information about the relevant coordinates of $f$ at different Fourier levels. 
Along the way, we give an algorithm which provides us with oracle access to a junta in the following sense:\footnote{A similar technique appeared in \cite[Section 5.1]{junta-coordinate-oracles} to sample two inputs on which the coordinate oracles agree. We note that our algorithm allows to specify the values the coordinate oracles attain.}
\begin{theorem}[Informal]
\label{theorem:implicit-junta-access-informal}
	Let $f: \pmone^n \to \pmone$, $\InfluentialOracles = \{g_1, \ldots, g_{k'}\}$ be a set of functions giving oracle access to a certain set of coordinates.
	Let $g$ be a function from $\pmone^{k'} \to [-1,1]$ defined by $g(x) = \E[f(y) | g_1(y)=x_1, \ldots, g_{k'}(y)=x_{k'}]$. Then $g$ can be computed by a randomized algorithm  that runs in expected time $\poly(k')$. 
\end{theorem}


We note that one can view this as an oracle access to the junta, without even figuring out the coordinates on which the junta depends. 
 More details on the ideas behind both algorithms can be found in \Cref{section:techniques}.

\subsection{Structure of this Paper}
\Cref{section:preliminaries} surveys some necessary preliminaries. \Cref{section:techniques} gives high level overviews of the techniques and ideas that go into the proofs of \Cref{theorem:improved-dmn} and \Cref{theorem:main-result}. 
 \Cref{section:improving-dmn} first describes how to get obtain ``oracle access'' to a junta (see \Cref{theorem:implicit-junta-access-informal}) using only oracles for relevant coordinates of the junta, and then provides all the details of the algorithm and proof for \Cref{theorem:improved-dmn}. Finally, \Cref{section:main-result} provides all the details of the algorithm and proof for \Cref{theorem:main-result}.

\section{Preliminaries}
\label{section:preliminaries}

Throughout the paper we adopt certain notation conventions. For a positive integer $n$, we denote by $[n]$ the set $\{1,\ldots,n\}$. For a distribution $\calD$, we denote that a random variable $x$ is sampled according to $\calD$ by $x \sim \calD$. In the case that $x$ is sampled uniformly at random from a set $S$, we will abuse notation slightly and write $x \sim S$. 
The binomial distribution with $n$ trials and probability $p$ per trial will be denoted $\Binomial(n,p)$. We denote the set $\{-1,1\}$ with the shorthand $\pmone$. For functions $f,g$ from $\pmone^n$ to $\pmone$ we define $\dist(f,g) = \Pr_{x\sim \pmone^n}[f(x) \neq g(x)]$: that is, the fraction of inputs on which $f$ and $g$ differ. For a set $S \subseteq [n]$ we will denote by $\BooleanHypercube{S}$ the set of possible assignments to the variables $\{x_i\}_{i \in S}$.

\subsection{Probability}
We recall the following Chernoff/Hoeffding bounds.
\begin{fact}
\label{fact:chernoff}
If $X_1,\ldots,X_N$ are independent random variables bounded in $[0,1]$ and $\bar{X} := \frac{1}{N} \sum_{i=1}^N X_i$, then we have
$$\Pr[|\bar{X} - \E[\bar{X}]| \geq \eta] \leq 2 \exp(-2N\eta^2),$$
Furthermore,  denoting by $p = \E[\bar{X}]$, we have 
$$\Pr[\bar{X} \leq p-\eta] \leq  \exp(-2N\eta^2),$$
$$
\Pr[\bar{X} \leq (1-\eta)p] \leq \left (\frac{e^{-\eta}}{(1-\eta)^{1-\eta}}\right)^{pN}\leq \exp\!\left(-\dfrac{\eta^2pN}{2}\right).
$$
\end{fact}

\subsection{Boolean Functions}
In this section we recall some tools in the analysis of Boolean functions. For a more thorough introduction to the field, we refer the reader to~\cite{AOBFbook}. 

For every subset $S\subseteq[n]$, we define the parity function on the bits in $S$, denoted by $\chi_S: \BooleanHypercube{n} \to \BooleanHypercube{}$ as $\chi_S(x) = \prod_{i \in S}x_i$.  It is a well-known fact that we can express uniquely any $f: \BooleanHypercube{n} \to \R$ as a linear combination of $\{\chi_S\}_{S\subseteq[n]}$: 
$$f(x) = \sum_{S\subseteq[n]}\hat{f}(S)\chi_S(x).$$
The coefficients $\{\hat{f}(S)\}_{S\subseteq[n]}$ are referred to as the Fourier coefficients of $f$, and can be calculated by $\hat{f}(S) = \E[f(x)\chi_S(x)]$. We say Fourier coefficients are on \textit{level} $s$ if they correspond to subsets of size $s$.

Given a function $f:\pmone^n \to \pmone$ and a coordinate $i \in [n]$, we define the \emph{influence} of the  $i$-th coordinate on $f$ to be
$$
\Inf_i[f] = \Pr_{x\sim \pmone^n}[f(x) \neq f(x^i)].
$$
It is a well-known fact (see, e.g., \cite[Theorem~2.20]{AOBFbook}) that $\Inf_i[f] = \sum_{S \ni i} \hat f(S)^2$. The latter definition naturally extends to  functions $f: \pmone^n \to \R$. We naturally extend this notion and define the \emph{low-degree influence} (up to level $k$) of coordinate $i$ on $f$ as
$$
\Inf_i^{\leq k}[f] = \sum_{S \ni i, |S| \leq k} \hat f(S)^2.
$$

For a set $T \subseteq [n]$ we define the {\em projection} of the function $f$ to $T$, denoted $f^{\subseteq T}$, as the partial Fourier expansion restricted to sets contained in $T$, i.e., $f^{\subseteq T}(x) = \sum_{S: S \subseteq T} \hat{f}(S) \chi_S(x)$. We observe that $f^{\subseteq T}$ depends only on coordinates in $T$ and that it can be alternatively defined as $f^{\subseteq T}(x) =\E_{y\sim \pmone^n}[f(y)|y_T=x_T]$. As suggested by the last identity, we also denote $f^{\subseteq T}$ by $f_{\avg, T}$.

In the regime of property testing, we will need a notion of ``closeness''  of functions.
\begin{definition}
For functions $f,g: \pmone^n \to \pmone$ and a set of functions $G$, all from $\pmone^n \to \pmone$ we say that
\begin{enumerate}
    \item $f$ is \emph{$\nu$-close} to $g$ if $\dist(f,g) \leq \nu$;
    \item $f$ is $\nu$-close to $G$ if $\min_{g \in G}\dist(f,g) \leq \nu$;
    \item $f$ and $g$ are $c$-correlated if $\E_{x \in \pmone^n}[f(x)g(x)]  = c$;
    \item $f$ and $G$ are $c$-correlated (denoted $\corr(f,G) = c$) if $\max_{g \in G} \E_{x \in \pmone^n}[f(x)g(x)]  = c$. 
\end{enumerate}
\end{definition}
In the paper, we will occasionally abbreviate the correlation between $f$ and $g$ as $\E[fg]$ when the domain is implied. Observe that when $f$ and $g$ are Boolean-valued (in $\pm 1$) we have $\E[fg] = 1-2\dist(f,g)$.
\begin{fact}
For functions $f,g: \pmone^n \to \R$, we have Plancheral's identity:
$$
\E_{x\sim\pmone^n}[f(x)g(x)] = \sum_{S \subseteq[n]} \hat f(S) \hat g(S)\;.
$$
When $f = g$, this fact is known as Parseval's identity.
\end{fact}
\begin{definition}
For a function $f: \pmone^n \to \R$ we define:
    $$ \W^{\leq k}[f] = \sum_{|S| \leq k} \hat f(S)^2\;.$$
\end{definition}
The definitions of $\W^{\geq k}[f]$, $\W^{=k}[f]$, and similar follow from a natural extension. Now, we define some classes of Boolean functions with properties that will be useful to us.
\begin{definition}[Junta]
Let $T\subseteq [n]$. A function $f: \pmone^n \to \R$ is called a \emph{junta} on $T$ if $f$ depends only on coordinates in $T$. I.e., there exists a function $g: \pmone^T \to \R$ such that $f(x) = g(x_T)$.
A function is called a~\emph{$k$-junta} if it is a junta on $T$ for some $T\subseteq[n]$ of size $k$. 
Following the notation of \cite{junta-coordinate-oracles}, we denote the class of $k$-juntas on $n$ inputs as $\Juntas{n}{k}$. We also denote $\Juntas{U}{k}$ as the set of $k$-juntas with inputs inside of $U$, and when $|U| = k$ then we often denote $\calJ_U:= \Juntas{U}{k}$ for brevity. 
\end{definition}

\begin{definition}[Dictator, Anti-Dictator]
The $i$-th dictator function is given by $\Dict_i(x) = x_i$, for $x \in \pmone^n$. The $i$-th antidictator function is simply the negation $-\Dict_i(x)$.
\end{definition}
\begin{claim}[Nearest $k$-junta on a Subset]
\label{claim:maximum-correlation-junta}
For a function $f:\BooleanHypercube{n} \to [-1,1]$ and a subset $T \subseteq [n]$, 
the Boolean-valued junta-on-$T$ most correlated with $f$ is given by
$$
\sgn(f_{\avg, T}(x)) = \sgn\left(\E_{y \in \BooleanHypercube{n}}[f(y)|y_T = x_T]\right).
$$
Furthermore, the correlation between $f$ and $\sgn(f_{\avg, T}(x))$ is simply
$\E_{x\sim \pmone^n}[|f_{\avg,T}(x)|]$.
\end{claim}
We keep the proof for this well-known claim for completeness.
\begin{proof}
Let $g: \pmone^n \to [-1,1]$ be any junta-on-$T$.
It suffices to show that $\E_x[f(x)g(x)] \le \E[f(x) \sgn(f_{\avg, T}(x))]$, as we do next.
Indeed, for any $g(x)$ that is a junta-on-$T$ we have $g(x) = g'(x_T)$ for some $g': \pmone^T \to [-1,1]$. Thus, we have
\begin{align*}
\E_{x\sim \pmone^n}[f(x)g(x)] &=
\E_{x\sim \pmone^n}[f(x)g'(x_T)]\\
&= \E_{x \sim \pmone^n}\left[g'(x_T) \cdot \E_{y \sim \pmone^n}[f(y)|x_T=y_T]\right]\\
&= \E_{x \sim \pmone^n}[g'(x_T) f_{\avg,T}(x)]\\
&\le  \E_{x \sim \pmone^n}[|f_{\avg,T}(x)|]\\
&= \E_{x \sim \pmone^n}[\sgn(f_{\avg,T}(x)) \cdot f_{\avg,T}(x)]\\
 &= \E_{x\sim \pmone^n}[f(x)\sgn(f_{\avg,T}(x))].\qedhere
\end{align*}

\end{proof}


A useful tool in  Boolean Function Analysis is the noise operator $T_{\rho}$.
For a vector $x \in \pmone^n$ we denote by $N_{\rho}(x)$ the distribution over vectors $y\in \pmone^n$ such that for each coordinate $i\in [n]$ independently $y_i = x_i$ with probability $(1+\rho)/2$ and  $y_i = -x_i$ otherwise (alternatively, $\E[x_i y_i] = \rho$).
For a function $f: \pmone^n \to \R$ we denote by $T_{\rho}f: \pmone^n \to \R$ the function defined by
$$T_{\rho}f(x) = \E_{y\sim N_{\rho}(x)}[ f(y)]$$
There's also a nice Fourier expression for the function $T_{\rho}f$ given by $T_{\rho}f(x) = \sum_{S \subseteq [n]} \hat{f}(S) \rho^{|S|}$.
We will need a simple fact about the noise operator.
\begin{fact}[\protect{\cite[Exercise 2.33]{AOBFbook}}]
\label{fact:noise-op-l1-norm}
For any function $f:\pmone^n \to \R$ and any $\rho \in [-1,1]$ we have that $\E[|T_{\rho}f|] \le \E[|f|]$.
\end{fact}

\subsection{Estimating Fourier Coefficients}
The following claim is a standard tool in many learning algorithms. It establishes that estimating Fourier coefficients of a Boolean function $f$ can be done with a few queries to $f$.
\begin{claim}[\protect{\cite[Proposition 3.39]{AOBFbook}}]
Suppose $f: \pmone^n \to \pmone$ and $S \subseteq [n]$ then there exists an algorithm that estimates $\hat{f}(S)$ up to additive error $\eps$ with probability at least $1-\delta$ that makes $O((1/\eps^2) \cdot \log(1/\delta))$ samples.
\end{claim}

The next claim generalizes the claim to a bounded function $f: \pmone^n \to [-1,1]$.
For that generalization, we need the definition of a randomized algorithm computing a bounded function $f$.
\begin{definition}[Randomized Algorithm for a Bounded Function]\label{def:randomized algorithm for bounded function}
Let $f: \pmone^n \to [-1,1]$ be a bounded function. We say that
algorithm A is a randomized algorithm for $f$ if on any fixed input $x$ algorithm A outputs a random bit $\bfy\in \pmone$ with $\E[\bfy] = f(x)$.
\end{definition}

\begin{claim}\label{claim:estimate Fourier non-Boolean}
Let $f: \pmone^n \to [-1,1]$, and 
let $A$ be a randomized algorithm for $f$.
Then, there exists an algorithm making $O((1/\eps^2) \cdot \log(1/\delta))$ calls to $A$ that estimates $\hat{f}(S)$ up to additive error $\eps$ with probability at least $1-\delta$.
\end{claim}
\begin{proof}[Proof Sketch]
We estimate $\hat{f}(S)$ by sampling $m = O((1/\eps^2) \cdot \log(1/\delta))$ uniformly random inputs $\bfx^{(1)}, \ldots, \bfx^{(m)}$, applying $A$ to  each of them to get random bits $(\bfy_1, \bfy_2, \ldots, \bfy_m)$, and taking the empirical mean of $\frac{1}{m} \sum_{i=1}^{m} \bfy_i \cdot \chi_S(\bfx^{(i)})$.
	Note that for each $i \in [m]$ we have that  $\bfy_i \cdot \chi_S(\bfx^{(i)})$ is a $\pmone$ random variable with expectation 
	$$\E_{\bfx^{(i)}, \bfy_i}[\bfy_i \cdot \chi_S(\bfx^{(i)})] = \E_{\bfx^{(i)}}\left[\E_{\bfy_i}[\bfy_i|\bfx^{(i)}]\cdot \chi_S(\bfx^{(i)})\right] = 
	\E_{\bfx^{(i)}}[f(\bfx^{(i)}) \cdot \chi_S(\bfx^{(i)})] = \hat{f}(S).$$
	The claim follows from \Cref{fact:chernoff}.
\end{proof}

\subsection{Random Restrictions}

\begin{definition}[Restriction]
\label{function-restriction}
Consider the class of  functions on $\pmone^n$.
A restriction is a pair $(J, z)$ where $J \subseteq[n]$, and $z \in \pmone^{\overline{J}}$.
Given a function $f:\pmone^n \to \mathbb{R}$, and a  restriction $(J,z)$, the {\sf restricted function} $f_{\overline{T} \to z}: \{\pm 1\}^{T} \to \mathbb R$ is defined by $f_{\overline{T} \to z}(x) = f(y)$ where $y_T = x$ and $y_{\overline{T}} = z$.
\end{definition}
\begin{definition}[$\delta$-Random Restriction]
\label{definition:random-restriction}
For $\delta \in [0,1]$ we say that $J$ is a {\sf $\delta$-random subset} of $S$ if it is formed by including each element independently with probability $\delta$, which we denote as $J \subseteq_\delta S$. 
A $\delta$-random restriction, denoted $(J,z) \sim \calR_\delta$, is sampled by taking $J$ to be a $\delta$-random subset $J$ on $[n]$, and taking $z$ to be a  uniformly random string in $\pmone^{\overline{J}}$.\end{definition}


Occasionally, we will abuse notation and think of $f_{\overline{T} \to z}$ as a function from $\pmone^n$ to $\pmone$ that ignores bits outside $T$. For example, $f_{\overline{T} \to z}: \pmone^n \to \pmone$ is given by $f_{\overline{T} \to z}(x) = f(x_T, z_{\overline{T}})$. Finally, we will use the following fact on random restrictions:

\begin{fact}[\protect{\cite[Corollary 3.22]{AOBFbook}}]
\label{fact:expected_square}
For a function $f: \pmone^n \to \R$ and sets $S\subseteq J \subseteq [n]$ we have 
$$
\E_{z \in \pmone^{\overline{J}}}[\hat{f_{\overline{J} \to z}}(S)^2] = \sum_{\substack{R \subseteq [n] ,  R\cap J= S}} \hat{f}(R)^2.
$$
\end{fact}

\section{Overview of Techniques}
\label{section:techniques}
Both of our algorithms rely on only having to consider a subset of influential coordinates, rather than all $n$ input variables. This is obtained using results from \cite{junta-coordinate-oracles}, and is discussed further in \Cref{section:improving-dmn}. For now, we simply assume that we are only dealing with $\poly(k,1/\eps)$ coordinates $\InfluentialCoords$. 
For  simplicity of presentation, we ignore dependence  on $\eps$, and focus only the dependence on $k$.  Thus, in this section, assume that $\eps$ is a small universal constant, e.g., $\eps = 0.01$.

\subsection{Techniques for Establishing \Cref{theorem:improved-dmn}}
\label{subsection:improved-dmn-techniques}

Our first result shows how to further reduce the number of coordinates we need to consider down to $O(k/\eps^2)$, while only losing at most $\epsilon$ amount of correlation with the maximally correlated $k$-junta.
In establishing \Cref{theorem:improved-dmn}, we first develop intuition behind a notion of normalized influence that we introduce next:

\begin{definition}[Normalized Influence]
\label{definition:normalized-influence}
Let $f: \pmone^n \to \R$.
We define the normalized influence of coordinate $i$ on $f$ as
$$\NormInf_{i}[f] = \sum_{ S \ni i}\frac{\hat{f}(S)^2}{|S|}.$$
We also naturally define the normalized influence below level $k$: 
$$\NormInf_i^{\leq k}[f]:= \sum_{\substack{|S| \leq k \\ S \ni i}}\frac{\hat{f}(S)^2}{|S|}.$$
\end{definition}

We note that while the term ``normalized influence'' is new, the quantity itself is not. It first appeared in a work of Talagrand \cite{Talagrand-zero-one} (expressed as $M(\Delta_i f)^2$) which generalized the famous KKL theorem \cite{DBLP:conf/focs/KahnKL88, DBLP:journals/eccc/KelmanKKMS20}, and subsequently appeared in followup works extending Talagrand's theorem to Schreier graphs \cite{OK-schreier}. As far as we know, this is the first use of this quantity in a learning or testing setting.

The next claim states that the sum of normalized influences of $f$ equals its variance.

\begin{claim}
For any function $f:\pmone^n\to\R$, we have that $\sum_i\NormInf_i[f] = \Var[f].$
\end{claim}
\begin{proof}
We have that
\[
\sum_{i\in [n]}\NormInf_i[f] = \sum_{i \in [n]}\sum_{ S \ni i}\frac{\hat{f}(S)^2}{|S|} = \sum_{\substack{ S \subseteq[n] \\ S \neq \emptyset}}\sum_{i \in S}\frac{\hat{f}(S)^2}{|S|} = \sum_{\substack{ S \subseteq[n] \\ S \neq \emptyset}}|S|\frac{\hat{f}(S)^2}{|S|} = \sum_{\substack{ S \subseteq[n] \\ S \neq \emptyset}}\hat{f}(S)^2 = \Var[f],
\]
where the last equality follows from Parseval's identity.
\end{proof}
\begin{remark}
We note that for a balanced Boolean function $f$ (that is, one where $\E_x[f(x)] = 0$) the normalized influences form a probability distribution on the coordinates $i$.
\end{remark}
The idea behind establishing \Cref{theorem:improved-dmn}  begins with the observation the these normalized influences can be thought of as defining a sub-probability distribution over the input coordinates of $f$, since these are non-negative numbers whose  sum is at most $1$. The weight assigned to coordinate $i$, similar to the regular influence, captures how important $i$ is to $f$, but assigns a higher relative weight to the coordinates with Fourier mass coming from the lower levels of the Fourier decomposition. 

The second important observation for us is that for any set $T$ of size at most $k$ we can write 
\begin{equation}
\label{eq:normed-infs}
    \sum_{i\in T} \NormInf_i^{\leq k}[f] = \sum_{i\in T} \sum_{\substack{|S| \leq k\\\emptyset \neq S \ni i}} \frac{\hat{f}(S)^2}{|S|} \geq \sum_{i\in T} \sum_{\substack{ S \subseteq T \\ S \ni i}}  \frac{\hat{f}(S)^2}{|S|} = \sum_{\emptyset \neq S \subseteq T}\hat{f}(S)^2.
\end{equation}

Intuitively, this shows that if some set of coordinates captures large amount of Fourier mass, then this same subset of coordinates also is very likely to be sampled by our sub-probability distribution defined by the normalized influences. Our idea follows this line of thought -- we get decent estimates for all of the normalized influences, and sample coordinates from this estimated distribution. 
Let $T$ be the \say{target set} of size $k$, i.e., the one for which the closest $k$-junta to $f$ is a junta on $T$.
Without loss of generality we can assume that $T$ captures constant fraction of the Fourier mass, meaning $\sum_{\emptyset \neq S \subseteq T}\hat{f}(S)^2\ge \Omega(1)$. Otherwise, the best correlation of $f$ with a $k$-junta is $o(1) < \eps$ and the task of $\eps$-accurately estimating the distance to the set of $k$-juntas becomes trivial.
Assuming  $T$ captures constant fraction of the Fourier mass,  \Cref{eq:normed-infs} tells us that we will sample $i \in T$ with constant probability mass. Thus, sampling from this distribution $O(k)$ times means we will have seen most of $T$ up to a small loss in correlation. 

To actually estimate these normalized influences, we apply a series of $\log 10k$ random restrictions to our function $f$ (first take $1$-random restrictions, then $1/2$-random restrictions, then $1/4$-random restrictions, and so on), and then show that summing $\hat{\Restrictedf}(\{i\})^2$ for each of these restrictions is sandwiched between $\NormInf_i^{\leq k}[f]$ and $\NormInf_i[f]$:
$$\frac{1}{2}\NormInf_i^{\leq k}[f] \leq \sum_{i=0}^{\log 10k} \E_{(J,z) \sim \calR_{2^{-i}}}\bigbracket{\hat{\Restrictedf}(\{i\})^2} \leq 2\NormInf_i[f].$$ 
This would allow us to effectively sample from a proxy distribution that still samples  $i \in T$  with constant probability.

We repeat the process iteratively, sampling coordinates one at a time, until we either sampled all of $T$ or sampled a subset $T' \subseteq T$ for which we have that the best junta on $T'$ is almost as correlated with $f$ as the best junta on $T$. Since the process samples a coordinate in $T$ with constant probability in each round, after $O(k)$ iterations we are likely to succeed, giving us a set $U$ of $O(k)$ coordinates that contains either $T$ or $T'$ (as above).
Finally, we show we can estimate, up to a small additive error, the best correlation of a junta-on-$U$ with $f$, given only approximate oracle access to the coordinates in $\InfluentialCoords$. By the above discussion the estimate we get is lower bounded by the best correlation with a $k$-junta up to a small additive error. It is also upper bounded (trivially) with the best correlation of $f$ with a $O(k)$-junta, since $|U| = O(k)$.

\subsection{Techniques for \Cref{theorem:main-result}}
\label{subsection:main-results-techniques}

A limitation of the algorithm we described in the previous subsection is that it only samples one coordinate at a time. 
In particular, suppose we want to find $T$ exactly, instead of a superset $U$ of $T$. Then, the naive algorithm would need to consider all subsets of $U$ of size $k$, estimating the best correlation with a junta on each of them. This gives a $\exp(O(k))$-query algorithm.
It would be nicer if we can devise a sampling algorithm that outputs, with constant probability, many coordinates of $T$ at a time. Such a sampling algorithm would reduce the number of possibilities for $T$ in the second stage.
In particular, consider the case that the nearest $k$-junta to $f$ had significant amount of Fourier mass on higher levels, say at level $\approx k$ or maybe $\approx \sqrt{k}$. In this case it would be nice to be able to sample from the Fourier distribution of $f$, that would give us a large subset of $T$ with constant probability. We note that sampling from the Fourier distribution of a Boolean function is easy for a quantum algorithm but hard for a randomized algorithm. Nevertheless, the (classical) algorithm we describe in this section takes inspiration from this, and samples subsets of size $\sqrt{k}$ according to the Fourier mass of $f$ above level $\sqrt{k}$ of each subset, in time and query complexity $\exp(\tilde{O}(\sqrt{k}))$. 

We will start with the preliminary that we have reduced to the case of only having to consider the coordinates in $\InfluentialCoords \subseteq [n]$ with $|\InfluentialCoords| \leq O(k/\eps^2)$, using our aforementioned algorithm from the previous section, incurring only a small additive loss in correlation with the closest $k$-junta. We start with the following definition that generalizes normalized influences of coordinates to normalized influences of sets of coordinates.
\begin{definition}
\label{definition:norm-influences-general}
For a given subset $U \subseteq [n]$, we define its normalized influence as follows:
$$\NormInf_U[f] := \sum_{S:\, U \subseteq S} \frac{\hat{f}(S)^2}{\binom{|S|}{|U|}}.$$
We also have the natural extension of $\NormInf_U^{\leq k}[f] = \sum_{S:\,|S|\le k, U \subseteq S} \frac{\hat{f}(S)^2}{\binom{|S|}{|U|}}$, analogous to \Cref{definition:normalized-influence}.
\end{definition}

This is a direct generalization of the quantity in \Cref{definition:normalized-influence}. In particular, we consider taking $|U| = \sqrt{k}$. Note there are $2^{\tilde{O}(\sqrt{k})}$ such $U$ within the coordinates in $\InfluentialCoords$, and we can think of these normalized influences as once again defining a sub-probability distribution over subsets of size $\sqrt{k}$. It likely does not sum to $1$, but rather sums to $\W_\InfluentialCoords^{\geq \sqrt{k}}[f] \leq 1$. We show that these normalized influences at exactly level $\sqrt{k}$ can once again be approximated to within a constant factor via a sequence of random restrictions to $f$:
$$\frac{1}{2}\NormInf_U^{\leq k}[f] \leq \sum_{i=0}^{2\sqrt{k}\log 10k} \E_{(J,z) \sim \calR_{p^i}}\bigbracket{\hat{\Restrictedf}(U)^2} \leq 3\NormInf_U[f],$$
where $p = \bigparen{1-\frac{1}{2\sqrt{k}}}$. For more details on this statement, see \Cref{theorem:ninf-ub-lb-set}.

We are now ready to outline the overall algorithm in \Cref{section:main-result}. Suppose $\specialset\subseteq \InfluentialCoords$ is the subset on which the nearest $k$-junta (within $\InfluentialCoords$) is defined. Our algorithm can then be broken down into two phases:
\begin{enumerate}
    \item[Phase 1.] We get a proxy for $\NormInf_U$ for all $|U| = \sqrt{k}$. This is achieved by performing a series of random restrictions to $f$.
    
    We consider these proxies as a distribution, and sample a constant (this constant is actually dependent on $\epsilon$, see \Cref{section:main-result} for details) number of subsets of size $\sqrt{k}$. With high probability, one of these is in our set of interest $\specialset$, provided $\specialset$ has a non-negligible amount of Fourier mass above level $\sqrt{k}$. 
    
     We don't know which of the subsets we sample are actually in $\specialset$, so we start a branching process. For each subset we sampled, we restrict $f$'s values in that subset, and recursively sample from sets of size $\sqrt{k}$ using the steps described above. Our branching process will have depth at most $\sqrt{k}$ since at each level we sample $\sqrt{k}$ new coordinates, and $\specialset$ can have at most $k$ relevant coordinates. This phase  of our algorithm produces $2^{\Tilde{O}(\sqrt{k})}$ possible subsets of our target set $\specialset$.
    \item[Phase 2.] With high probability, one of the branches in the above process will have captured most of the coefficients of $T$ that are relevant above level $\sqrt{k}$ on the Fourier spectrum. Each branch of this process represents a different possibility for what $\specialset$ may be, so for each branch we randomly restrict $f$ so that the coordinates sampled in that branch are fixed, which effectively moves most of the mass of $T$ to levels below $\sqrt{k}$. 
    We then estimate all the Fourier coefficients of this restricted $f$ below level $\sqrt{k}$, allowing us to get an estimate for the closest $k$-junta on any subset using these estimated coefficients. 
    Each estimation of a Fourier coefficient requires $2^{\tilde{O}(\sqrt{k})}$-queries to estimate to the desired accuracy, and there are $2^{\tilde{O}(\sqrt{k})}$ Fourier coefficients to estimate, so overall we make at most $2^{\tilde{O}(\sqrt{k})}$ queries. From there, for each possible subset of $B \subseteq \specialset$ outputted by phase one, we brute force over all possible subsets of size $k$ containing $B$, estimating the correlation $f$ has with the closest $k$-junta on that subset using our estimated Fourier coefficients. This last step takes exponential time in $k$. We emphasize that while our runtime is exponential in $k$, our query complexity is only exponential in $\tilde{O}(\sqrt{k})$.
\end{enumerate} 
In the entire above explanation, we have eliminated the dependence on $\epsilon$ for simplicity. We also only consider $\specialset$ for conceptual and analytic simplicity -- in reality, we have no idea what $\specialset$ is, and indeed it is exactly what we are looking for. Therefore, more work must be done in order to show that we do not accidentally pick the wrong set, for which our estimates may be inaccurate. To get around this subtle issue, we further apply a noise operator in order to ensure that the significant parts of $f$ lie below level roughly $\sqrt{k}$. We discuss this further in \Cref{subsection:phase-two-lower-levels}.

\section{Finding a Small(er) Set of  Influential Coordinate Oracles}
\label{section:improving-dmn}
In this section, we detail the process of constructing oracles to coordinates with large low-degree influence. We expand upon the techniques in \cite{junta-coordinate-oracles}, reducing the number of coordinates one needs to consider to produce a highly correlated $k$-junta (assuming one exists). 
\subsection{Approximate Oracles to Influential Coordinates}
\label{subsection:approx-oracles}
In this subsection we outline and generalize the methods used by \cite{junta-coordinate-oracles} to achieve oracle access to coordinates with large low-degree infuence in $f$. We start with the following definitions from their paper, repeated here for clarity:

\begin{definition}[\protect{\cite[Def.~3.1]{junta-coordinate-oracles}}]
Let $\InfluentialOracles$ be a set of functions mapping $\{ \pm 1\}^n$ to $\{ \pm 1\}$. We say that $\InfluentialOracles$ is an oracle for the coordinates in $\InfluentialCoords$ if 
\begin{itemize}
    \item for every $g \in \InfluentialOracles$, there is some $i \in \InfluentialCoords$ such that $g = \pm \Dict_i$; and
    \item for every $i \in \InfluentialCoords$, there is some $g \in \InfluentialOracles$ such that $g = \pm \Dict_i$.
\end{itemize}
In other words, $\InfluentialOracles$ is an oracle for $\InfluentialCoords$ if $\InfluentialOracles = \{ \Dict_i : i \in \InfluentialCoords\}$ ``up to sign".
\end{definition}

However, it is not tractable to achieve perfect access to such oracles, so we have to settle for the following weaker notion of approximate oracles:
\begin{definition}[\protect{\cite[Def.~3.2]{junta-coordinate-oracles}}]
\label{definition:approx-oracles}
Let $\InfluentialOracles$ be a set of functions mapping $\{ \pm 1\}^n$ to $\{ \pm 1\}$. We say that $\InfluentialOracles$ is an $\nu$-oracle for the coordinates in $\InfluentialCoords$ if

\begin{itemize}
    \item for every $g \in \InfluentialOracles$, there is some $i \in \InfluentialCoords$ such that $g$ is $\nu$-close to $\pm \Dict_i$; and
    \item for every $i \in \InfluentialCoords$, there is exactly one $g \in \InfluentialOracles$ such that $g$  is $\nu$-close to $\pm \Dict_i$; and
    \item For every $g \in \InfluentialOracles$, and $\delta > 0$, there is a randomized algorithm that compute $g(x)$ correctly on any $x \in \pmone^n$ with probability at least $1-\delta$, using $\poly(k, \log \frac{1}{\delta})$ queries to $f$.
\end{itemize}
\end{definition}

 Lemma 3.6 in \cite{junta-coordinate-oracles} establishes that we can achieve access to a set $\InfluentialOracles$ of approximate oracles to $\InfluentialCoords \supseteq \{i\,: \, \Inf_i^{\leq k}[f] \geq \epsilon^2/k\}$ of bounded size. 

More specifically, we have the following corollary:

\begin{corollary}[\protect{\cite[Lemma~3.6]{junta-coordinate-oracles}}]
\label{corollary:get-oracles}
With $\poly(k, \frac{1}{\epsilon}, \log \frac{1}{\delta})\cdot \frac{1}{\nu}$ queries to $f$, we can gain access to an approximate oracle set $\InfluentialOracles$ in the sense that for every coordinate $i$ such that $\Inf_i^{\leq k}[f] \geq \frac{\epsilon^2}{k}$, there exists a $g \in \InfluentialOracles$ such that $g$ is $\nu$-close to $\pm \Dict_i$ with probability at least $1-\delta$. Furthermore, $|\InfluentialOracles|\leq \poly(k, \frac{1}{\epsilon}, \log(1/\delta))$.
\end{corollary}

For our purposes, we take $\nu = 0.1$ and $\delta = 2^{-\poly(k, \frac{1}{\epsilon})}$ in all our  algorithms. Since we will  make much fewer than $2^{\poly(k/\eps)}$-many queries to the coordinate oracles,  we can assume that all of our oracles are indeed $\nu = 0.1$ close to dictators/anti-dictators, since by a union bound this is true with high probability.

It is important to note that we do not have a description of which coordinates are influential: from an information theoretic standpoint this would require query complexity dependent on $n$. What we do have is oracle access to these coordinates in the sense that for all $i$ such that $\Inf_i^{\leq k}[f] \geq \epsilon^2/k$, there exists $g_i \in \InfluentialOracles$
such that $g_i(x) \approx \pm \Dict_i(x)$, that is, $\InfluentialOracles$ contains dictators or anti-dictators to every influential coordinate. Using simple techniques of local correction we can simplify this: we need only consider dictators to each coordinate in the oracle. Also, we can convert closeness on average $x$ to high probability correctness for all $x$ (i.e., a worst-case guarantee).\begin{lemma}
\label{lemma:localcorrect}
Suppose $f$ is $\nu$-close to $\pm \Dict_i$. For any $x \in \BooleanHypercube{n}$, $\LocalCorrect(f, x)$ samples a random $y \sim \BooleanHypercube{n}$ and outputs $f(y)f(x\cdot y)$, where $x \cdot y$ is pointwise multiplication. Then,$$\forall{x}: \Pr_{y\sim \pmone^n}[\LocalCorrect(f,x) \neq \Dict_i(x)] \leq 2\nu.$$

\end{lemma}
\begin{proof}
Suppose that $f$ is $\nu$ close to $\Dict_i$. Then we have $\Pr_{y \sim \pmone^n}[f(y) \neq \Dict_i(y)] \leq \nu$, and since $x\cdot y$ has the same distribution as $y$, $\Pr_{y \sim \pmone^n}[f(x\cdot y) \neq \Dict_i(x \cdot y)] \leq \nu$. Let $A$ be the event that $f(y) \neq \Dict_i(y)$ and let $B$ be the event that $f(x\cdot y) \neq \Dict_i(x \cdot y)$. Clearly if $\LocalCorrect(f,x) \neq \Dict_i(x)$ then at least one of $A$ and $B$ must have occurred (since $\Dict_i(x) = \Dict_i(x\cdot y) \cdot \Dict_i(y)$). Thus, by the union bound, we have
$$
\Pr_{y \sim \pmone^n}[\LocalCorrect(f,x) \neq \Dict_i(x)]\leq \Pr[A \cup B] \leq \Pr[A] + \Pr[B] \leq 2\nu
$$
A similar argument shows that if $f$ is $\nu$ close to $-\Dict_i$, then $\LocalCorrect(f,x)$ is not equal to $(-\Dict_i(y))(-\Dict_i(x \cdot y)) = \Dict_i(y)\Dict_i(x \cdot y) = \Dict_i(x)$ with probability at most $2\nu$. 
\end{proof}

Given a noisy black box computing $h$ which is $\nu$-close to $g = \pm \Dict_i$, local correction will compute $\Dict_i$ with high probability, on every input $x$. Critically, we can treat potentially faulty $\pm \Dict_i$ oracles as correct $\Dict_i$ oracles provided suitably many repetitions.
\begin{corollary}
\label{corollary:localcorrect}
If $f$ is $\nu$-close to $\pm \Dict_i$ for $\nu=0.1$, then repeating $\LocalCorrect(f,x)$ independently $\poly(k,1/\eps)$ times and taking the majority outcome results in an incorrect value for $\Dict_i(x)$ with probability at most $2^{-\poly(k,1/\eps)}$.
\end{corollary}

\begin{proof}
Clear from applying the first bound in~\Cref{fact:chernoff} with $N = O(\poly(k/\eps))$ and  $\eta =(1-2\nu-0.5) = 0.3$ in this case.
\end{proof}

We also show that restricting our attention to $\InfluentialCoords$ we have not lost more than $\eps$ in the best correlation of $f$ with a $k$-junta. This is proved in the following claim.
\begin{claim}
\label{claim:improved-junta-corr}
    Let $f: \{\pm 1 \}^n \to \{\pm 1\}$ and let $g: \{\pm 1\}^n \to \{\pm 1\}$ be a $k\text{-junta}$ on $U$.
    Let $\tau >0$. Take $$S= \left\{i \in U\; \middle\vert\; \Inf_i^{\leq k}[f] \geq \tfrac{\tau^2}{k}\right\}$$
    Then, there is a junta on $S$ with correlation at least $\E[fg] - \tau$ with $f$.
\end{claim}

\begin{proof}
To prove this claim, we define a function on the set $S$ such that the loss in correlation is at most $\tau$. Consider:
$$g'(x) = g_{\avg,S}(x) =  \E_y[g(y)|y_S = x_S]$$
First, we note $g'$ is a function over only the variables in $S$. Second, it is bounded in $[-1,1]$, so it is not quite Boolean, but it can be randomized rounded to a Boolean function, with the expected correlation with $f$ equaling $\E[fg']$. Thus, it suffices to show that $\E[fg'] \ge \E[fg]-\tau$ to deduce that there exists a randomize rounding of $g'$ to a Boolean function $g''$ with $\E[f g''] \ge \E[fg]-\tau$. We also recall that 
\[
\hat{g'}(T) = \begin{cases}
\hat{g}(T) & \text{ if } T \subseteq S \\
0 & \text{otherwise}
\end{cases}
\]
We thus have:
\begin{align*}
    \left| \E[fg] - \E[fg']\right| = \bigg|  \sum_{\substack{T\nsubseteq S \\ T \subseteq U}} \hat{f}(T)\hat{g}(T)\bigg|
    &\leq \sqrt{\sum_{\substack{T\nsubseteq S \\ T \subseteq U}} \hat{f}(T)^2} 
    \leq \sqrt{\sum_{i\in U\setminus S} \sum_{\substack{T \ni i \\ T \subseteq U}} \hat{f}(T)^2} \\
    &\leq \sqrt{\sum_{i\in U\setminus S}\Inf_i^{\leq k}(f)}
    \leq \sqrt{k \cdot \frac{\tau^2}{k}} = \tau \qedhere
\end{align*}
\end{proof}

Finally, the below corollary summarizes what we have achieved in this section.
\newcommand{\ideal}{\mathsf{ideal}}

   \begin{corollary}
    \label{corollary:approximate-to-exact-oracles}
    With $\poly(k, \frac{1}{\epsilon}, \log \frac{1}{\delta})$ queries to $f$, we can gain access to an approximate oracle set $\InfluentialOracles$ for a set of coordinates $\{i: \Inf_i^{\leq k}\geq \frac{\eps^2}{k}\} \subseteq \InfluentialCoords \subseteq [n]$. Moreover, these coordinates and oracles satisfy the following properties.  
    \begin{itemize}
        \item For every coordinate $i \in \InfluentialCoords$, there exists a $g \in \InfluentialOracles$ such that $g$ is $0.1$-close to $\Dict_i$ with probability at least $1-\delta$.
        \item $\dist(f, \Juntas{n}{k}) - \dist(f, \Juntas{\InfluentialCoords}{k}) \leq \eps$.
        \item $|\InfluentialCoords| \leq \poly(k,1/\eps, \log (1/\delta))$.
        \item For any algorithm $A$ that uses at most $q$ queries to $\InfluentialOracles$, we can use $\LocalCorrect$ from $\Cref{lemma:localcorrect}$ with error $\delta/q$ to assume that we actually have perfect access to each coordinate oracle, up to an additive loss of $\delta$ in confidence and a multiplicative overhead of $\poly(\log(q/\delta))$ in query complexity.
    \end{itemize}
\end{corollary}
\begin{proof}
The first and the third bullet point follow from \Cref{corollary:get-oracles}. The second bullet point follows from \Cref{claim:improved-junta-corr}. To achieve the last point, we can use \Cref{corollary:localcorrect} every time we make a ``query'' to an oracle in our algorithm. Thus every ``query" to an oracle $g\approx \pm \Dict_i$ at $x$ involves $\poly(\log(q/\delta))$ many repetitions of $\LocalCorrect(g,x)$, which results in an incorrect value with probability at most $\delta/2q$, as noted above. Recall that \Cref{corollary:get-oracles} guarantees that we can output $g(x)$ correctly with probability $1-\delta/2q$ with only a $\poly(k,\log(q/\delta))$ queries to $f$. 
Since we only ever make at most $q$ queries to our coordinate oracles, we can assume that $\LocalCorrect(g,x) =  \Dict_i(x)$ in all queries.
This happens with probability at least $1-\delta$ by the union bound.
\end{proof}
Therefore, for the rest of this paper, we will assume that we have oracle access to \textit{exact dictators}.

\subsection{Implicit Access to an Underlying Junta}
An important consequence of having coordinate oracles is that it allows us to reduce the input size of the function dramatically. Suppose $f: \pmone^n \to \pmone$ and we have $\mathcal{D} = \{g_1, \ldots, g_{{k'}}\}$ are randomized algorithms that for any $x\in \pmone^n$ output $g_i(x) = \Dict_{j_i}(x) = x_{j_i}$.
We have that $j_1, \ldots, j_{k'} \in [n]$ are a set of ${k'}$ distinct coordinates.

Let $U = \{j_1, \ldots, j_{k'}\}$.
We want to get access to the following function:
$g(x_1, \ldots, x_{k'}) = \E[f(y)| y_{j_1}=x_1, y_{j_2} = x_2, \ldots, y_{j_{k'}} = x_{{k'}}]$. More precisely, given $x_1, \ldots, x_{{k'}}$ we want to sample uniformly from all $y\in \pmone^n$ that satisfy $y_{j_1}=x_1, y_{j_2} = x_2, \ldots, y_{j_{k'}} = x_{{k'}}$ and apply $f$ on this $y$.

The following algorithm that runs in $\poly(k, \log(1/\delta))$ time samples  $y$ from such a distribution. 

\begin{algorithm}[H]
\DontPrintSemicolon
  \KwInput{$f$ (target function), $\InfluentialOracles = \{g_1, \ldots, g_{k'}\}$ (coordinate oracles), $(x_1, \ldots, x_{{k'}})\in \pmone^{{k'}}$}
  \KwOutput{A vector $y \in \pmone^n$ with $(g_1(y), \ldots, g_{k'}(y)) = (x_1, \ldots, x_{k'})$}%
  Sample $y \sim \pmone^{n}$ and let $z \in \pmone^{k'}$ be the vector of evaluations of $\{g_1, \ldots, g_{k'}\}$ on $y$; \\
   \While{$z \neq x$}{
       \Repeat{$\dist(x, z') <\dist(x,z)$}{
    Let $y'$ be a copy of $y$, but flip each bit independently with probability $\frac{1}{{k'}}$;
     \\
     Let $z'$ be the vector of evaluations of $\{g_1, \ldots, g_{k'}\}$ on $y'$;  }
     $y = y'$; \\
     $z = z'$;
    }
\Return{$y$}
\caption{Sampling a uniformly random input consistent with the oracles' values}
\label{alg:sampling}
\end{algorithm}

\begin{theorem}
\Cref{alg:sampling}  with probability $1-\delta$ runs in  time $\poly({k'}, \log(1/\delta))$.
\end{theorem}
\begin{proof}
We focus on the number of iterations of the inner repeat loop.
Given $(y,z)$ with $z\neq x$ we analyze the time it takes to find a $(y',z')$ with $\dist(z',x)<\dist(z,x)$.
Since $x \neq z$ without loss of generality we can assume that $x_1 \neq z_1$. 
To get $(y',z')$ with $\dist(z',x)<\dist(z,x)$, 
it suffices to sample a vector $y'$ with $y'_{j_1}=x_1$ and $y'_{j_2} = y_{j_2}, y'_{j_3} = y_{j_3}, \ldots, y'_{j_{k'}} = y_{j_{k'}}$. Indeed, since we are flipping each coordinate with probability $1/{k'}$ the probability of sampling such a $y'$ is exactly $1/{k'} \cdot (1-1/{k'})^{{k'}-1} \ge 1/(e{k'})$.
Thus, we get that the runtime of the repeat loop is stochastically dominated by a geometric random variable with success probability $1/(e{k'})$. Thus  with probability at least $1-\delta/{k'}$, it finishes after  $O({k'} \cdot \log({k'}/\delta))$ iterations. We run the inner repeat loop at most ${k'}$-times, thus by union bound, with probability at least $1-\delta$ the entire process end after at most  $O({k'}^2 \cdot \log({k'}/\delta))$ executions of line~5.
We note that execution line~5 actually requires ${k'}$ queries to $g_1, \ldots, g_{k'}$, each of them takes $\poly(k) = \poly({k'})$ time. thus overall, with probability at least $1-\delta$, our algorithm run in time $\poly({k'}, \log(1/\delta))$.
\end{proof}

\begin{theorem}
\Cref{alg:sampling}  samples uniformly from the set of inputs $\{y': (g_1(y'), \ldots, g_{k'}(y')) = (x_1, \ldots, x_{k'}))\}$.
\end{theorem}
\begin{proof}
Let $U = \{j_1, \ldots, j_{k'}\}$ be the set of coordinates for which $\{g_1, \ldots, g_{k'}\}$ are	oracles to.
\Cref{alg:sampling}  certainly samples a vector $y$ with $y_{j_1} = x_1, \ldots, y_{j_{k'}} = x_{k'}$.
We want to show additionally that \Cref{alg:sampling} samples $y_{\overline{U}}$ uniformly at random. In fact, at any point in the algorithm the distribution over $y_{\overline{U}}$ is uniform. This is clearly true in the first step where $y\sim \pmone^n$, and remains true along the algorithm as we apply independent noise to coordinates in $\overline{U}$ and decide whether to apply the noise or not according to the value of $y_U$ which is independent of $y_{\overline{U}}$.
\end{proof}

We will consider algorithms computing non-Boolean function like $g = f_{\avg, S}$ for some subset $S \subseteq [n]$. Note that $g$ is a function whose range in $[-1,1]$, but not necessarily a Boolean function.
\begin{theorem}[Formal version of \Cref{theorem:implicit-junta-access-informal}]
\label{theorem:implicit-junta-access}
	Let $f: \pmone^n \to \pmone$, $\InfluentialOracles = \{g_1, \ldots, g_{k'}\}$ be a set of coordinate oracles.
	Let $g$ be a function from $\pmone^{k'} \to [-1,1]$ defined by $g(x) = \E[f(y) | g_1(y)=x_1, \ldots, g_{k'}(y)=x_{k'}]$. Then $g$ has a randomized algorithm in the sense of \Cref{def:randomized algorithm for bounded function} computing it that runs in expected time $\poly(k')$. 
\end{theorem}
\begin{proof}
	Given $x = (x_1, \ldots,x_{k'})$ apply \Cref{alg:sampling} on $f$, $\InfluentialOracles$ and $x$ to get a vector $y\in \pmone^n$. Return $f(y)$.
	It is clear that since $y$ is a uniform input subject to $g_1(y)=x_1, \ldots, g_{k'}(y)=x_{k'}$ that our algorithm is a randomized algorithm for $g$.
\end{proof}

\subsection{Influential Coordinate Oracles}

 As above, denote as $\InfluentialCoords$ the  superset of the low-degree influential coordinates of $f$, and $\InfluentialOracles$ as the set of approximate oracles to said coordinates, obtained via \Cref{corollary:get-oracles} with parameter $\nu = 0.1$. As we discussed in \Cref{subsection:approx-oracles}, we assume (with a small loss in error probability, and a small multiplicative factor on query complexity) that we have exact access to dictators for each influential coordinate. We work towards proving the following improved version of a corollary that appeared in \cite{junta-coordinate-oracles}:

 The idea will be to take $\InfluentialOracles$, a set of $k' = \poly\left(k, \frac{1}{\epsilon}\right)$ coordinate oracles, and somehow ``prune" it down to a set  $\InfluentialOracles'$ of at most $O(\frac{k}{\epsilon^2})$ coordinate oracles, such that that the loss in the most correlated junta on this smaller set of coordinates is at most $\epsilon$ 

$$\max_{g \in \Juntas{\InfluentialOracles}{k}} \E[fg] - \max_{g \in \Juntas{\InfluentialOracles'}{k}} \E[fg] \leq \epsilon.$$
\subsection{Reducing the Number of Oracles to Consider}

Starting with a set of $\poly(k/\eps)$ set of oracles $\InfluentialOracles$ for a set $\InfluentialCoords$ containing the influential coordinates of $f$,
our goal in this section is to prune the number of oracles to $O(k/\eps^2)$ in a way that incurs only a small loss in correlation with the nearest $k$-junta. 
\cite{junta-coordinate-oracles} achieved their theorem by noting that applying a standard noise operator to $f$ did not affect its proximity to the nearest $k$-junta significantly, while also guaranteeing that at most $\frac{k^2}{\epsilon^2}$ coordinates could have large influence. They then were able to estimate the influence of every coordinate in $\InfluentialOracles$ despite only having (approximate) oracle access to the influential coordinates, and thus were able to determine which oracles were actually oracles to influential coordinates, of which there were less than $k^2/\epsilon^2$.

Our approach, as explained at a high level in \Cref{section:techniques}, is to estimate the normalized influence of each coordinate in $\InfluentialCoords$, which is done via a sequence of random restrictions to $f$.
In words, the below algorithm estimates for each coordinate $i \in \InfluentialCoords$ the quantity $\lambda_i^{\approx 2^d} =\E_{(J,z)\sim \calR_{2^{-d}}}[ \hat{\Restrictedf}(\{i\})^2]$,
where $(J,z) \sim \calR_{2^{-d}}$ parameterize a $2^{-d}$-random restriction to $f$. Then, $\lambda_i$ is defined to be sum over a series of random restrictions $d=0,...,\log 10k$ of $\lambda_i^{\approx 2^d}$. The core idea of our algorithm is that this sum over Fourier coefficients on the first level of restricted versions of $f$ is a proxy for $\NormInf_i[f]$. In other words, we have the following theorem:
\begin{theorem}
\label{theorem:ninf-ub-lb}
Let $f: \pmone^{{k'}} \to \R$, where ${k'} = |\InfluentialOracles|$.
Let $i \in [{k'}]$.
Let 
$$\lambda_i[f] = \sum_{m=0}^{\log(10k)} \lambda_i^{\approx 2^{m}}[f],
\qquad \text{where}\qquad 
\lambda_i^{\approx 2^{m}}[f] = 
 \E_{(J,z) \sim \calR_{2^{-m}}}[\hat{\Restrictedf}(\{i\})^2].$$
Then, $\frac{1}{2}\NormInf_i^{\leq k}[f] \le \lambda_i[f] \le 2 \NormInf_i[f]$.
\end{theorem}

We postpone the proof of \Cref{theorem:ninf-ub-lb} to \Cref{subsection:proof-of-ninf-ub-lb}.
The definition of $\lambda_i$ naturally gives rise to an algorithm for estimating $\lambda_i$ that we present next. The algorithm would return for each $i \in [{k'}]$ an estimate $\tilde{\lambda}_{i}$ that would be close to $\lambda_i$ with high probability. 

\begin{algorithm}[H]
\DontPrintSemicolon
  \KwInput{$f:\pmone^{k'} \to [-1,1]$ along with randomized algorithm A  computing $f$ (recall Def.~\ref{def:randomized algorithm for bounded function}). Parameters  $1-\delta$ (confidence),  $\eps$ (additive error) and $k$.}
  \KwOutput{Estimates $(\tilde{\lambda_1}, \ldots, \tilde{\lambda_{k'}})$ for $(\lambda_1, \ldots, \lambda_{k'})$.}
  Let $m= \poly(k, k', 1/\eps, \log(1/\delta))$\\
  Initialize $\tilde{\lambda_i} = 0$ for all $i\in [k']$;\\
  \For{$d=0$ \KwTo $\log 10k$}{
      Initialize $\tilde{\lambda}_i^{\approx 2^d} = 0$ for all $i \in [k']$;\\
    
    \RepTimes{$m$} {
        Let $(J,z) \sim \calR_{2^{-d}}$ be a $2^{-d}$-random restriction.\\
        
        Estimate $\hat{\Restrictedf}(\{j\})$ for all $j \in J$ up to additive error $\frac{\eps}{6\log(10k)}$ with probability $1-\delta/\poly(k,k',m)$ using \Cref{claim:estimate Fourier non-Boolean} and algorithm A.\\
        Denote by $\tilde{\Restrictedf}(\{j\})$ the estimated Fourier coefficient.\\
        Update $\tilde{\lambda}_j^{\approx 2^d} = \tilde{\lambda}_j^{\approx 2^d} +  \tilde{\Restrictedf}(\{j\})^2$ for all $j \in J$.
    }
    Let $\tilde{\lambda}_i^{\approx 2^d} = \tilde{\lambda}_i^{\approx 2^d}/m$ for all $i\in [k']$;
  }
  Let $\tilde{\lambda_i} = \sum_d \tilde{\lambda}_i^{\approx 2^d}$;\\
\Return{$(\tilde{\lambda_1}, \tilde{\lambda_2}, \ldots, \tilde{\lambda_{k'}})$}
\caption{Estimating $\lambda_i$}
\label{alg:lambda_estimate}
\end{algorithm}

\begin{lemma}
\label{lemma:lambda-i-correctness}
With probability at least $1-\delta$ we have that for all $i\in [k']$ it holds that 
$|\tilde{\lambda}_i -\lambda_i|\le \eps$.
\end{lemma}
\begin{proof}
If $j \notin J$ the Fourier coefficient of $\hat{\Restrictedf}$ is $0$ and so our estimate is correct in that case. In the case $j \in J$, each estimation of the Fourier coefficient is correct up to additive error $\eta = \eps/6 \log(10k)$ with probability at least $1-\delta/\poly(k,k',m)$.
Thus, we get that $\tilde{\Restrictedf}(\{j\})^2 = (\hat{\Restrictedf}(\{j\}) \pm \eta)^2 = \hat{\Restrictedf}(\{j\})^2  \pm 2\eta|\hat{\Restrictedf}(\{j\})| \pm \eta^2 = \hat{\Restrictedf}(\{j\})^2 \pm 3\eta$.
Furthermore, we have that $\E_{(J,z)\sim \calR_{2^{-d}}}[\hat{\Restrictedf}(\{j\})^2] = \lambda_j^{\approx 2^d}$, thus by \Cref{fact:chernoff} we have that the empirical mean of $m = \poly(1/\eps,\log(k), \log(k'), \log(1/\delta))$ copies of  $\tilde{\Restrictedf}(\{j\})^2$ is within additive error $\eps/(2 \log(10k))$ from $\lambda_j^{\approx 2^d}$ with probability at least  $1-\delta/(k' \log(10k))$. By union bound, all these estimates are within the error bound, and we get that $|\tilde{\lambda}_j^{\approx 2^d} - \lambda_j^{\approx 2^d}| \le 3\eta + \eps/(2 \log(10k)) \le \eps/(\log(10k))$. Overall, we get that $|\tilde{\lambda}_j - \lambda_j| \le \eps$ for all $j \in[k']$ with probability at least $1-\delta$.
\end{proof}

With \Cref{alg:lambda_estimate} in hand, we are ready to present the pruning procedure.

\begin{algorithm}[H]
\DontPrintSemicolon
  
  \KwInput{$f$ (target function), $\InfluentialOracles$ (influential coordinate oracles, where $\InfluentialOracles$ are oracles for $\InfluentialCoords$). Parameters $\epsilon$ and $\delta$.}
  \KwOutput{A subset $\InfluentialOracles' \subseteq \InfluentialOracles$ of size $O(\frac{k}{\epsilon^2})$ such that we lose at most $\epsilon$ in correlation with $f$.} 
  
  Initialize $\InfluentialOracles' = \emptyset$;\\
  Let $m = O((k+ \log(1/\delta))/\epsilon^2)$\\
  \RepTimes{$m$}{
	Let $\{g_1, \ldots, g_{k'}\} = \InfluentialOracles - \InfluentialOracles'$, and $\{g_{k'+1}, \ldots, g_{|\InfluentialOracles|}\} = \InfluentialOracles'$\\
	Sample $z \in \pmone^{|\InfluentialOracles'|}$.
	Let $f': \pmone^{k'} \to \R$ be the function defined by 
	$$f'(x_1, \ldots, x_{k'}) = \E_{y\sim \pmone^n}
	[f(y) | g_1(y)=x_1, \ldots, g_{k'}(y) = x_{k'}, g_{k'+1}(y) = z_1, \ldots, g_{k'+|\InfluentialOracles'|}(y)=z_{|\InfluentialOracles'|}].$$ 
	and let A be the randomized algorithm for $f'$ from \Cref{theorem:implicit-junta-access}.

	 Apply \Cref{alg:lambda_estimate} on $f'$ using the randomized algorithm A for $f'$ with confidence $1-\frac{\delta}{2m}$ and accuracy $\frac{\eps^2}{48\cdot |\InfluentialCoords|}$ $\implies \tilde{\lambda}=(\tilde{\lambda}_{1}, \ldots, \tilde{\lambda}_{k'})$. \\
	 Let our distribution $P$ be defined by $\tilde{\lambda}$, normalized appropriately.\\
    Sample $i \sim P$, and add $g_i$ to $\InfluentialOracles'$.\\
  }
    
\Return{$\InfluentialOracles' $} \\

\caption{Reduce Number of Oracles}
\label{alg:reduce-number-of-oracles}
\end{algorithm}

\begin{lemma}
\label{lemma:DMN-impovement-correctness}
With probability at least $1-\delta$, \Cref{alg:reduce-number-of-oracles} returns a set of oracles $\InfluentialOracles'$ to a subset of coordinates $\InfluentialCoords' \subseteq \InfluentialCoords$, such that
$$\max_{g \in \Juntas{\InfluentialCoords}{k}} \E[fg] - \max_{g \in \Juntas{\InfluentialCoords'}{k}} \E[fg] \leq \epsilon.$$
\end{lemma}

To prove \Cref{lemma:DMN-impovement-correctness}, which tells us our algorithm succeeds and directly implies \Cref{theorem:improved-dmn}, we will need a few more lemmas.

We denote the event $\calE$ that in the entire execution of \Cref{alg:reduce-number-of-oracles} all $\tilde{\lambda}_i$ were $\eps^2/(48\cdot |\InfluentialCoords|)$ close to the real $\lambda_i$.
We note that by union bound this event happens with probability at least $1-\delta/2$.

Suppose $T$ is the (unknown) set of $k$ oracles for which the best-$k$ junta approximating $f$ is a junta on $T$. We want to show that our algorithm either samples all the coordinates in $T$, or it samples a subset $T'$ of $T$ that captures all but $\eps^2/4$ of the Fourier mass of $f$ on $T$.

\begin{claim}
\label{claim:Sprime-correctness}
Assume the event $\calE$ happens. Then, with probability at least $1-\delta/2$, after $m$ iterations, we will have either:
 \begin{enumerate}
     \item sampled $i$ for all $i\in T$, our target set;
     \item sampled $i$ for all $i \in T' \subseteq T$, where 
     $\sum_{S \subseteq T'} \hat{f}(S)^2 \geq \sum_{S \subseteq T} \hat{f}(S)^2 - \epsilon^2/4$. 
 \end{enumerate}
\end{claim}

\begin{proof}
In each iteration, assume we have not yet satisfied either items. Let $V$ be the subset of coordinates in $T$ that we have  not yet sampled.
Let $T' = T \setminus V$. By assumption, $$\epsilon^2/4<\sum_{S\subseteq T} \hat{f}(S)^2- \sum_{S\subseteq T'} \hat{f}(S)^2 =  \sum_{S\subseteq T : S \cap V \neq \emptyset} \hat{f}(S)^2.$$

Let $\InfluentialCoords'' = \InfluentialCoords\setminus \InfluentialCoords'$. We have that $|\InfluentialCoords''| = k'$.
Now note that up to relabeling of coordinates 
$f'$ from \Cref{alg:reduce-number-of-oracles} is the same as $(f_{\avg, \InfluentialCoords})_{\InfluentialCoords' \to z}$,
where $z$ was randomly chosen.
For brevity, denote by $f_z = (f_{\avg, \InfluentialCoords})_{\InfluentialCoords' \to z}$. Note that for any fixed $z$, $f_z$ is a function that depends only on the coordinates in $\InfluentialCoords''$.

By \Cref{fact:expected_square}, we have 
\begin{align}
\E_z\left[ \sum_{\emptyset \neq S \subseteq V}\hat{f_z}(S)^2\right] 
= \sum_{R:  \emptyset \neq (R\cap \InfluentialCoords'') \subseteq V}\hat{f_{\avg, \InfluentialCoords}}(R)^2 =
\sum_{\substack{R\subseteq \InfluentialCoords : \\ \emptyset \neq (R\cap \InfluentialCoords'') \subseteq V}}\hat{f}(R)^2 
&\ge 
\sum_{R\subseteq T : R \cap V \neq \emptyset} \hat{f}(R)^2 \nonumber\\&> \eps^2/4.\label{eq:fz}
\end{align}

Next, by applying \Cref{theorem:ninf-ub-lb}, for any fixed $z$, we have
$$
\sum_{i\in V} \lambda_i[f_z]\ge \frac{1}{2} \sum_{i\in V} \NormInf_i^{\le k}[f_z]\ge \frac{1}{2} \sum_{\emptyset \neq S \subseteq V} \hat{f_z}(S)^2.
$$
By the assumption that  $\calE$ happens, the $\tilde{\lambda_i}$ are $\frac{\eps^2}{48\cdot |\InfluentialCoords|}$-accurate, and we get that 
$$
\sum_{i \in V} \tilde{\lambda_i}[f_z] \ge \frac{1}{2} \sum_{\emptyset \neq S \subseteq V}\hat{f_z}(S)^2 - \frac{\eps^2}{48\cdot |\InfluentialCoords|} \cdot |V|
\ge \frac{1}{2} \sum_{\emptyset \neq S \subseteq V}\hat{f_z}(S)^2 - \frac{\eps^2}{48}.
$$
On the other hand by applying \Cref{theorem:ninf-ub-lb} again we see that $$\sum_{i\in \InfluentialCoords''} \lambda_i[f_z] \le 2\cdot \sum_{i\in \InfluentialCoords''} \NormInf_i[f_z]= 2\cdot \Var[f_z]\le 2$$ 
and thus 
$\sum_{i\in \InfluentialCoords''} \tilde{\lambda}_i[f]  \le 2 + k'\cdot \frac{\eps^2}{48\cdot |\InfluentialCoords|} \le 2 + \frac{\eps^2}{48} \le 3$ (under the assumption that $\calE$ happens).
Overall, the probability to sample an element from $V$ is at least 
$$
\frac{1}{3} \cdot \left(\frac{1}{2} \sum_{\emptyset \neq S \subseteq V}\hat{f_z}(S)^2 - \frac{\eps^2}{48}\right) 
= \frac{1}{6} \sum_{\emptyset \neq S \subseteq V}\hat{f_z}(S)^2 - \frac{\eps^2}{3\cdot 48}
$$
By taking  expectation over $z$, and using \Cref{eq:fz} we see that the probability to sample an element from $V$ overall is at least 
$$
\E_{z}\left[\frac{1}{6} \sum_{\emptyset \neq S \subseteq V}\hat{f_z}(S)^2 - \frac{\eps^2}{3\cdot 48}\right] 
\ge  \frac{1}{6} \cdot \frac{\eps^2}{4} - \frac{\eps^2}{3\cdot 48} > \frac{\eps^2}{30}.
$$

We get that in each iteration as long as we don't satisfy Items~(1) and (2) above, we sample an element from $i\in T$ with probability at least $\eps^2/30$. By repeating the process $m = O(\frac{k+\log(1/\delta)}{\eps^2})$ times we would sample all of $T$, or get stuck at some $T'$ satisfying Item~(2), with probability at least $1-\delta/2$, using \Cref{fact:chernoff}.
\end{proof}

Next, we show that finding $T'$ is almost as good as finding $T$ in the sense that the best correlation by juntas-on-$T'$ with $f$ is up to small additive error the best correlation by juntas-on-$T$ with $f$.
\begin{lemma}
\label{lemma:small-mass-loss-okay}
Suppose we have some subset $T$ such that $\sum_{S\subseteq T} \hat{f}(S)^2  = c$, and we then identified a subset $T' \subseteq T$ such that $\sum_{S \subseteq T'} \hat{f}(S)^2 \geq c - \frac{\epsilon^2}{4}$. Then 
$$\left | \max_{g \in \mathcal{J}_{T,k}} \E[fg] -  \max_{g \in \mathcal{J}_{T',k}} \E[fg] \right | \leq \epsilon$$
\end{lemma}
\begin{proof}
We know that $\argmax_{g \in \mathcal{J}_{T,k}}E[fg] = \sgn(f_{\avg,T})$ and similarly $\argmax_{g \in \mathcal{J}_{T',k}}E[fg] = \sgn(f_{\avg,T'})$. 
Then we have that 
\begin{align*}
\left | \max_{g \in \mathcal{J}_{T,k}} \E[fg] -  \max_{g \in \mathcal{J}_{T',k}} \E[fg] \right | &= \E[f(x)(\sgn(f_{\avg,T}(x_T)) - \sgn(f_{\avg,T'}(x_{T'}))] \\
&= \E_{x_T}\left [ \E_{x_{\overline{T}}} [f(x_T, x_{\overline{T}})] \left (\sgn(f_{\avg,T}(x_T)) - \sgn(f_{\avg,T'}(x_{T'})\right )\right ] \\
&= \E_{x_T}\left [ f_{\avg,T}(x_T) \left (\sgn(f_{\avg,T}(x_T)) - \sgn(f_{\avg,T'}(x_{T'})\right )\right ] \\
&\leq 2 \E_{x_T}\left [ \left |f_{\avg,T}(x_T) - f_{\avg,T'}(x_{T'})\right |\right ] \tag{Since $z(\sgn(z) - \sgn(z')) \leq 2 |z - z'|$ for all $z,z'\in \R$} \\
&\leq 2 \sqrt{\E_{x_T}\left [ \left (f_{\avg,T}(x_T) - f_{\avg,T'}(x_{T'})\right )^2\right ]} \\
&= 2 \sqrt{\sum_{S \subseteq T} \hat{f}(S)^2 - 2 \sum_{S \subseteq T'} \hat{f}(S)^2 + \sum_{S \subseteq T'} \hat{f}(S)^2} \\
&\leq 2 \sqrt{\frac{\epsilon^2}{4}} = \epsilon\;.\qedhere
\end{align*}
\end{proof}

\begin{proof}[Proof of \Cref{lemma:DMN-impovement-correctness}]
Let $g$ be the $k$-junta that maximizes $\E[fg]$ among all $k$-juntas on $\InfluentialCoords$.
Let $T$ be the set of  variables on which $g$ depends.
By Claim~\ref{claim:Sprime-correctness} we either sample oracles to all of $T$ or to a subset $T'$ for which $$\sum_{S \subseteq T'} \hat{f}(S)^2 \geq \sum_{S \subseteq T} \hat{f}(S)^2 - \epsilon^2/4.$$
In the second case, by \Cref{lemma:small-mass-loss-okay}, we incur a loss in correlation of at most $\epsilon$ with our nearest $k$-junta. In the first case, we lose no correlation with the closest $k$-junta, and by a union bound our probability of failure is at most $\delta$.
\end{proof}
The above concludes the proof of \Cref{lemma:DMN-impovement-correctness}. Finally, \Cref{theorem:improved-dmn} is implied by \Cref{lemma:DMN-impovement-correctness}, as shown below.

\begin{theorem}[Theorem~\ref{theorem:improved-dmn}, restated]
Let $\epsilon>0$, $k \in \N$, and $k' = C(k/\eps^2)$ for some universal constant $C$. Then, there exists an algorithm that given $f, k, \eps$ makes at most $poly(k, 1 / \eps)$ queries to $f$  and returns a number $\alpha$ such that with probability at least $0.99$
\begin{enumerate}
	\item $\alpha \le \max_{g \in \mathcal{J}_{n,k'}} \E[fg]+O(\epsilon)$
	\item $\alpha \ge \max_{g \in \mathcal{J}_{n,k}} \E[fg]-O(\epsilon)$
\end{enumerate}
\end{theorem}

\begin{proof}
Set $\delta = 2^{-\poly(k, 1/\eps)}$.
We first apply \Cref{corollary:get-oracles} from \cite{junta-coordinate-oracles}. This gives us $\poly(k, \frac{1}{\epsilon}, \log(1/\delta)) = \poly(k/\eps)$ coordinate oracles $\InfluentialOracles$ to coordinates $\InfluentialCoords$ that includes all coordinates $i$ with $\Inf_i^{\leq k}[f] \geq \frac{\epsilon^2}{k}$.
By \Cref{claim:improved-junta-corr} we see that $$\max_{g \in \mathcal{J}_{\InfluentialCoords,k}} \E[fg]\ge \max_{g \in \mathcal{J}_{n,k}}\E[fg] - \eps$$
Next, we apply \Cref{alg:reduce-number-of-oracles} to get a subset $\InfluentialOracles' \subseteq \InfluentialOracles$ to coordinates $\InfluentialCoords' \subseteq \InfluentialCoords$ such that with high probability
 $$\max_{g \in \mathcal{J}_{\InfluentialCoords',k}} \E[fg]\ge \max_{g \in \mathcal{J}_{\InfluentialCoords,k}}\E[fg] - \eps$$
We take $\alpha$ to be the estimation of  the correlation of the best junta on $\InfluentialCoords'$ with $f$.
By \Cref{claim:maximum-correlation-junta} we have that $\max_{g \in \mathcal{J}_{\InfluentialCoords'}}\E[fg] = \E[|f_{\avg, \InfluentialCoords'}(x)|]$.
To estimate the latter, we use a randomized algorithm that computes $f_{\avg, \InfluentialCoords'}$ given by \Cref{theorem:implicit-junta-access}.
We randomly sample $O(1/\eps^2)$ many values for $x$ and estimate for each of them $|f_{\avg, \InfluentialCoords'}(x)|$ up to additive error $\eps/2$ via the randomized algorithm with expected value $f_{\avg, \InfluentialCoords'}(x)$.

Assume that $\alpha$ is a $\eps$-additive approximation to $\max_{g \in \mathcal{J}_{\InfluentialCoords'}} \E[fg]$.
In this case, we claim that $\alpha$ satisfies both items from the theorem's statement.
Indeed, 
\begin{enumerate}
	\item $\alpha \le \max_{g \in \mathcal{J}_{\InfluentialCoords'}} \E[fg] + \eps \le \max_{g \in \mathcal{J}_{n,k'}} \E[fg] + \eps$.
	\item $\alpha \ge \max_{g \in \mathcal{J}_{\InfluentialCoords'}} \E[fg] - \eps \ge \max_{g \in \mathcal{J}_{\InfluentialCoords',k}} \E[fg] -\eps \ge  \max_{g \in \mathcal{J}_{\InfluentialCoords,k}}\E[fg] - 2\eps
	\ge \max_{g \in \mathcal{J}_{n,k}}\E[fg] - 3\eps$.
\end{enumerate}

Next, we analyze the number of queries of our algorithm.
Obtaining the initial set of coordinate oracles $\InfluentialOracles$ takes $\poly(k, 1/\eps, \log(1/\delta)) = \poly(k,1/\eps)$ queries.
Then, we go on to run \Cref{alg:reduce-number-of-oracles} that makes $m = O((k+ \log(1/\delta))/\epsilon^2)$ iterations, each making $\poly(k, 1/\eps,\log(1/\delta))$ queries.
Next, to estimate $\E[|f_{\avg, \InfluentialCoords'}(x)|]$ we require $\poly(1/\eps)$ samples from randomized algorithm for $f_{\avg, \InfluentialCoords'}(x)$ each such sample translate to $\poly(k, 1/\eps)$ samples to $f$. Finally, we note that each ``query'' to an oracle incurs an overhead of $\poly(\log(k,1/\eps))$ queries to $f$ along with an $o(1)$ additive loss in confidence by \Cref{corollary:approximate-to-exact-oracles}.
Overall, we make $\poly(k, 1/\eps)$ queries.
\end{proof}

\subsection{Proof of \Cref{theorem:ninf-ub-lb}}
\label{subsection:proof-of-ninf-ub-lb}
We now present the proof of \Cref{theorem:ninf-ub-lb}.
\begin{proof}[Proof of \Cref{theorem:ninf-ub-lb}]
We express $\lambda_i$ in terms of the Fourier spectrum of $f$. Using \Cref{fact:expected_square},
\begin{align*}\lambda_i &= \sum_{m=0}^{\log(10k)} \sum_{S: S \ni i} \hat{f}(S)^2 \cdot \Pr_{J \subseteq_{2^{-m}} [{k'}]}[ S \cap J = \{i\}]
\\ &= \sum_{m=0}^{\log(10k)} \sum_{S: S \ni i} \hat{f}(S)^2 \cdot \Pr_{J \subseteq_{2^{-m}} [{k'}]}[ |S \cap J| = 1] \cdot \frac{1}{|S|}
\\&=\sum_{S: S \ni i} \frac{\hat{f}(S)^2}{|S|} \cdot \sum_{m=0}^{\log(10k)} \Pr_{J \subseteq_{2^{-m}} [{k'}]}[ |S \cap J| = 1]\end{align*}
It therefore suffices to show that for any non-empty set $S$ such that $|S| \leq k$ it holds that
\begin{equation}\label{eq:1}
\frac{1}{2} \le \sum_{m=0}^{\log(10k)} \Pr_{J^{(m)} \subseteq_{2^{-m}} [{k'}]}[ |S \cap J^{(m)}| = 1] \le 2\;.
\end{equation}
From which it is clear that $\lambda_i \le 2 \cdot \sum_{S: S \ni i} \frac{\hat{f}(S)^2}{|S|}  = 2 \cdot \NormInf_i[f]$ and similarly $\lambda_i \ge \frac{1}{2}\sum_{\substack{ S \ni i,\\ |S|\le k}} \frac{\hat{f}(S)^2}{|S|}  =   \frac{1}{2}\NormInf_i^{\leq k}[f]$.

We move to prove \Cref{eq:1}.
The first observation is that an equivalent way to sample $J^{(m)} \subseteq_{2^{-m}}[{k'}]$ is to sample $m$ independent set $J^{(m)}_1, \ldots, J^{(m)}_{m} \subseteq_{1/2}[{k'}]$ and take their intersection $J^{(m)} = J^{(m)}_1 \cap \cdots \cap J^{(m)}_{m}$. Furthermore,  by linearity of expectation
$$\sum_{m=0}^{\infty} \Pr_{J^{(m)} \subseteq_{2^{-m}} [{k'}]}[ |S \cap J| = 1]
= \sum_{m=0}^{\infty}\E_{\substack{J_1^{(m)} \subseteq_{1/2}[{k'}], \\J_2^{(m)} \subseteq_{1/2}[{k'}], \\ \ldots}} \left [ \mathbbm{1}_{|S \cap J_1^{(m)}  \cap \cdots \cap J_m^{(m)}| = 1}\right ]
=\E_{\substack{J_1 \subseteq_{1/2}[{k'}], \\J_2 \subseteq_{1/2}[{k'}], \\\ldots}}\left[ \sum_{m=0}^{\infty} \mathbbm{1}_{|S \cap J_1 \cap \cdots \cap J_m| = 1}\right]$$
which in essence means that the choices for $J_1^{(1)}, J_1^{(2)}, \ldots$ can be the same set $J_1$, and similarly for any $J_i$.

To analyze the latter expectation, we note that it can be described as the expected value of the following random process:\\
\begin{algorithm}[H]
\DontPrintSemicolon
$X \leftarrow 0$\\
  \For{$i=1, 2, \ldots, \log(10k)$} {
    \If{$S=\emptyset$}{ halt!;
    }
    \If{$|S|=1$}{ increment $X$;
    }
     Sample $J_i \subseteq_{1/2} [{k'}]$; \\
     $S \leftarrow S \cap J_i$; \\
	}
\end{algorithm}

It therefore suffices to show that the expected value of the above random process is bounded in $[1/2, 2]$.
In the analysis, we consider also the infinite horizon process that keeps on going until $S= \emptyset$.
We observe that the expected values of both processes depend only on the size of the initial $S$ from symmetry.
For any $t \in \{0,1, \ldots, {k'}\}$, denote by $F_{t}$ the expected value of the infinite horizon process starting with a set $S$ of size $t$.
For the finite horizon process with $i$ iterations, we let the expected value be denoted by $F_t^{(i)}$.
We observe that $F_0 = 0$, and furthermore that $F_1 = 2$ since starting from a set of size $1$ the random variable $X$ would behave like geometric random variable with $p=1/2$.
Similarly, 
$F_1^{(i)} = 2 - \frac{1}{2^{i - 1}}$  as it is the minimum of $i$ and a geometric random variable with $p=1/2$.

Furthermore, for the infinite horizon process, we observe that we have the following recurrence $$F_t = \sum_{a=0}^{t} \frac{\binom{t}{a}}{2^{t}} \cdot F_{a},$$
for $t\ge 2$ or equivalently
$$F_t\cdot (1-2^{-t}) = \sum_{a=0}^{t-1} \frac{\binom{t}{a}}{2^{t}} \cdot F_{a}.$$
We show by induction that $1/2< F_t^{(\log 10k)} \le 2$ for $t\ge 1$.
The base case $t=1$ was discussed above.
Applying the induction hypothesis we have 
$$
F_t\cdot (1-2^{-t}) = \sum_{a=0}^{t-1} \frac{\binom{t}{a}}{2^{t}} \cdot F_{a}
\le \sum_{a=0}^{t-1} \frac{\binom{t}{a}}{2^{t}} \cdot 2
\le (1-2^{-t}) \cdot 2.$$
Dividing both sides by $(1-2^{-t})$ gives the  inequality $F_t \le 2$, which implies that $F_t^{(\log 10k)}\leq 2$.

For the lower bound, we consider the indicator random variable $Y_t^{(i)}$, where $t = |S|$, which equals 1 if $|S|=1$ at some point during the above process before iteration $i$. We note that $Y_t^{(\log 10k)}$ is a lower bound for the value of $X$ in the finite horizon process, and $Y_t$ is a lower bound for the value of $X$ at the end of the infinite horizon process. First, we claim that $\E[Y_t] = \Pr[Y_t = 1] \geq 2/3$ for all $t \geq 1$. The base case of $t=1$ is certainly true, and we also have, similar to before, that
\begin{align*}
    \E[Y_t]\cdot (1-2^{-t}) &= \sum_{a=0}^{t-1} \frac{\binom{t}{a}}{2^t}\E[Y_a] \\
    &\ge 0\cdot \frac{1}{2^t} + 1\cdot\frac{t}{2^t} +\frac{2}{3}\cdot \underbrace{\sum_{a=2}^{t-1} \frac{\binom{t}{a}}{2^t}}_{1 - \frac{2+t}{2^t}} \\
    &= \frac{2}{3} + \frac{t - \frac{2}{3}(2+t)}{2^t} \geq \frac{2}{3} + \frac{t/3 - 4/3}{2^t} \ge \frac{2}{3} - \frac{2/3}{2^t} = \frac{2}{3} \cdot (1-2^{-t})
    \end{align*}
which holds for all $t \geq 2$, and thus $\Pr[Y_t = 1] \geq 2/3$. However, this only holds for the infinite horizon random process. Let $A$ be the event that $S = \emptyset$ by iteration $\log 10k$, and note that $\Pr[A] = \Pr[\Binomial(|S|, \frac{1}{10k}) = 0] \geq \Pr[\Binomial(k, \frac{1}{10k}) = 0]  = \left( 1-\frac{1}{10k}\right )^{k}  \geq 1-\frac{k}{10k} = 0.9$. Finally, we claim that for all $t\geq 2$ we have that $\Pr[Y_t^{(\log 10t)}]\geq 1/2$. Note that for $Y_t$ to happen, it must be the case that either $\overline{A}$ happens or $Y_t^{(\log 10t)}$ happens. Thus, by a union bound
\begin{align*}
\tfrac{2}{3} \le \Pr[Y_t = 1] \le 
    \Pr[Y_t^{(\log 10t)}=1]  + \Pr[\overline{A}] \le     \Pr[Y_t^{(\log 10t)}=1] + 0.1\;,
\end{align*}
which  implies  $\Pr[Y_t^{(\log 10t)}=1] >1/2$. Finally, $F_t^{(\log 10t)} \geq \Pr[Y_t^{(\log 10t)}=1] > 1/2$ as desired.
\end{proof}

\section{A $2^{\Tilde{O}(\sqrt{k})}$-query Tolerant Junta Tester}
\label{section:main-result}

In this section, we prove \Cref{theorem:main-result}. Throughout this section, we  assume that we already applied \Cref{alg:reduce-number-of-oracles} to reduce the number of coordinate oracles to $O(k/\eps^2)$. We denote by $\InfluentialOracles$ the set of oracles we get, and by $\InfluentialCoords \subseteq [n]$ the set of coordinate to which they are oracles to.
Suppose that the best $k$-junta approximation of $f$ is a junta-on-$T$, for a 
set $T \subseteq \InfluentialCoords$  of size $k$. We call $T$ the ``target set''. Note that $T$ is unknown to the algorithm, and in fact, identifying $T$ (or a close approximation to $T$) from all subsets of  size $k$ of $\InfluentialCoords$ is the crux of the problem.

We start with the observation that if we were somehow able to identify all of the variables of $\specialset$ that capture most of the Fourier mass above level $\fourierCutoff$, then we could simply restrict $f$ by randomly fixing these variables, leaving us with the task of identifying the best $k$-junta approximation of $f$, given that we know the best $k$-junta has most its Fourier mass below level $\fourierCutoff$. For the latter case, there are only $\binom{|\InfluentialCoords|}{\fourierCutoff}$ Fourier coefficients to estimate, and estimating these to sufficient accuracy allows one to estimate the the correlation $f$ has with any subset $U\subseteq \InfluentialCoords$ such that $|U| \leq k$.

We are now ready to present the details of the algorithm. The algorithm can be  broken down into two main steps. First, we find, with high probability, a set $B \subseteq \specialset$ that captures almost all Fourier mass of $T$ above level $\fourierCutoff$. This first step, which we call ``phase one'', closely resembles the techniques in \Cref{section:improving-dmn} in that we utilize a series of random restrictions to estimate normalized influences. The main difference is that rather than considering normalized influences of individual coordinates, we now consider normalized influences of sets of size $\fourierCutoff$. The goal of phase one is to produce at least one subset $B$ of our target set $\specialset$ which effectively captures most of the Fourier mass within $T$ above level $\fourierCutoff$. Once we have done that, we have reduced to the scenario of the closest $k$-junta to $f$ having most of its Fourier mass below level $\fourierCutoff$, which can be solved via estimating all of the Fourier coefficients below level $\fourierCutoff$. 

\subsection{Phase One: The Higher Levels}

First, we prove an analogous theorem to \Cref{theorem:ninf-ub-lb}, which relates $\lambda_U[f]$ to $\NormInf_U[f]$ for all $U$:
\begin{theorem}
\label{theorem:ninf-ub-lb-set}
Let $f: \pmone^\ell \to \R$.
Let $U \subseteq [\ell]$, where $\ell = |\InfluentialOracles|$ and $|U| \leq k$.
Let 
$$\lambda_U[f] = \sum_{m=0}^{2|U|\log(10k)} \lambda_U^{\approx p^{-m}}[f],\qquad \text{where}\qquad 
\lambda_U^{\approx p^{-m}}[f] = \E_{(J,z) \sim \calR_{p^{m}}}[\hat{\Restrictedf}(U)^2]$$

for $p = 1-\frac{1}{2|U|}$.
Then, $\frac{1}{2}\cdot \NormInf_U^{\leq k}[f] \le \lambda_U[f] \le 3\cdot \NormInf_U[f]$.
\end{theorem}
Again, we postpone the proof of this to the end of this section in \Cref{subsection:proof-ninf-ub-lb-set}.
The definition of $\lambda_U[f]$ is naturally algorithmic, and therefore we can design the following algorithm to approximate the values of $\lambda_U[f]$ for all sets $U$ of size $\fourierCutoff = \sqrt{\eps k}$.

\begin{algorithm}[H]
\DontPrintSemicolon
  \KwInput{$f:\pmone^{k'} \to [-1,1]$ along with a randomized algorithm A computing $f$ (recall Def.~\ref{def:randomized algorithm for bounded function}). Parameters  $1-\delta$ (confidence),  $\eps$ (additive error) and $k$.}
  \KwOutput{Estimates $\{\Tilde{\lambda}_U\}_{|U| = \fourierCutoff}$ for $\{\lambda_U\}_{|U| = \fourierCutoff}$.}
  Let $m= \poly(k, k', 1/\eps, \log(1/\delta))$\\
  Initialize $\tilde{\lambda}_U = 0$ for all $U \subseteq [k']$, $|U| = \fourierCutoff = \sqrt{\eps k}$\\
  Let $p = \left(1-\frac{1}{2\fourierCutoff}\right)$\\
  \For{$d=0$ \KwTo $2\fourierCutoff\log 10k$}{
      Initialize $\tilde{\lambda}_U^{\approx p^{-d}} = 0$ for all $U \subseteq [k']$ such that $|U| = \fourierCutoff$\\
    
    \RepTimes{$m$} {
        Let $(J,z) \sim \calR_{p^d}$ be a $p^{d}$-random restriction.\\
        
        Estimate $\hat{\Restrictedf}(U)$ for all $U \subseteq J$ of size $\fourierCutoff$ up to additive error $\frac{\eps}{12\fourierCutoff\log(10k)}$ with probability $1-\frac{\delta}{\binom{k'}{\fourierCutoff}m\cdot 2\fourierCutoff\log(10k)}$ using \Cref{claim:estimate Fourier non-Boolean} and algorithm A. 
        Denote by $\tilde{\Restrictedf}(U)$ the estimated Fourier coefficient.\\
        Update $\tilde{\lambda}_U^{\approx p^{-d}} = \tilde{\lambda}_U^{\approx p^{-d}} +  \tilde{\Restrictedf}(U)^2$ for all $U \subseteq J$ of size $\fourierCutoff$.
    }
    Let $\tilde{\lambda}_U^{\approx p^{-d}} = \tilde{\lambda}_U^{\approx p^{-d}}/m$ for all $U\subseteq J$ of size $\fourierCutoff$;
  }
  Let $\tilde{\lambda}_U = \sum_d \tilde{\lambda}_U^{\approx p^{-d}}$;\\
\Return{$\{\Tilde{\lambda}_U\}_{|U| = \fourierCutoff}$}
\caption{Estimating $\lambda_U$'s}
\label{alg:lambda_estimate_set}
\end{algorithm}

\begin{lemma}
\label{lemma:lambda-U-correctness}
With probability at least $1-\delta$ we have that for all $U \subseteq [k']$ of size $\fourierCutoff$ it holds that 
$|\tilde{\lambda}_U -\lambda_U[f]|\le \eps$.
\end{lemma}

\begin{proof}
This proof closely follows that of \Cref{lemma:lambda-i-correctness}.
If $U \not\subseteq J$ the Fourier coefficient of $\hat{\Restrictedf}(U)$ is $0$ and so our estimate is correct in that case. In the case $U \subseteq J$, each estimation of the Fourier coefficient is correct up to additive error $\eta = \frac{\eps}{12\fourierCutoff\log(10k)}$ with probability at least $1-\delta/\exp(k,k',m)$.
Thus, we get that $\tilde{\Restrictedf}(U)^2 = (\hat{\Restrictedf}(U) \pm \eta)^2 = \hat{\Restrictedf}(U)^2  \pm 2\eta|\hat{\Restrictedf}(U)| \pm \eta^2 = \hat{\Restrictedf}(U)^2 \pm 3\eta$.
Furthermore, we have that $\E_{(J,z)\sim \calR_{p^{d}}}[\hat{\Restrictedf}(U)^2] = \lambda_U^{\approx p^{-d}}$, thus by \Cref{fact:chernoff} we have that the empirical mean of $m = \poly(1/\eps,\poly(k), \poly(k'), \log(1/\delta))$ copies of  $\tilde{\Restrictedf}(U)^2$ is within additive error $\eps/(4\fourierCutoff\log(10k))$ from $\lambda_U^{\approx p^{-d}}$ with probability at least  $1-\frac{\delta}{\binom{k'}{\fourierCutoff}m\cdot 2\fourierCutoff\log(10k)}$. 
By union bound, all these estimates are within the error bound, and we get that
$$\left|\tilde{\lambda}_U^{\approx p^{-d}} - \lambda_U^{\approx p^{-d}}\right| \le 3\eta + \eps/(4 \fourierCutoff\log(10k)) \le \eps/(2\fourierCutoff\log(10k)).$$
Overall, we get that $|\tilde{\lambda}_U - \lambda_U[f]| \le \eps$ for all $|U| = \fourierCutoff$ with probability at least $1-\delta$.
\end{proof}

Since we are sampling sets of size $\fourierCutoff$, we need to sample at most $k/\fourierCutoff = \sqrt{k/\eps} =: \branchProcessDepth$ distinct subsets of $\specialset$ of size $\fourierCutoff$ in order to capture all the potential mass of $\specialset$ above level $\fourierCutoff$. 

\begin{algorithm}[H]
\DontPrintSemicolon
  
  \KwInput{$f$ (target function), $\InfluentialOracles$ (where $\InfluentialOracles$ are coordinate oracles for $\InfluentialCoords$) a current depth $t$, a current subset $\InfluentialOracles' \subseteq \InfluentialOracles$ of coordinate oracles, $\epsilon$, $\delta$}
  \KwOutput{Return collection of subsets of $\InfluentialOracles$ of size at most $k$.} 
  Let $\branchProcessDepth = k/\fourierCutoff = \sqrt{k/\eps}$\\
  Let $r = O(1/\eps^2)$ and $\ell = 2(r+1)^{3\branchProcessDepth+\log(2/\delta)}$\\ 
  \tcc{$r+1$ is the branching factor, and $\ell$ is an upper bound on the number of nodes in the branching process (the process depth is $3\branchProcessDepth+\log(2/\delta)$).}
  \If{$t = 3\branchProcessDepth+\log(2/\delta)$ \textbf{or} $|\InfluentialOracles'| > k -\fourierCutoff$} {
    \Return{$\{\InfluentialOracles'\}$}
  }
Let $\{g_1,...,g_{k'}\} = \InfluentialOracles - \InfluentialOracles'$ and $\{g_{k'+1},...,g_{|\InfluentialOracles|}\}$ where $k' = |\InfluentialOracles| - |\InfluentialOracles'|$\\
Sample $z \in \BooleanHypercube{|\InfluentialOracles'|}$. Let $f': \BooleanHypercube{k'}\to \R$ be the function defined by 
$$f'(x_1, \ldots, x_{k'}) = \E_{y\sim \pmone^n}
[f(y) | g_1(y)=x_1, \ldots, g_{k'}(y) = x_{k'}, g_{k'+1}(y) = z_1, \ldots, g_{|\InfluentialOracles|}(y)=z_{|\InfluentialOracles'|}], $$ 
and let A be the randomized algorithm for $f'$ from \Cref{theorem:implicit-junta-access}.

 Apply \Cref{alg:lambda_estimate_set} on $f'$ using the randomized algorithm A for $f'$ with confidence $1 - \frac{\delta}{2\ell}$ and accuracy $\frac{\eps^2}{48\cdot \binom{|\InfluentialOracles|}{\fourierCutoff}}$ $\implies \tilde{\lambda}=\{\tilde{\lambda}_U\}_{|U| = \fourierCutoff}$. \\
 Let our distribution $P$ be defined by $\tilde{\lambda}$, normalized appropriately\\
 Sample $M_1,...,M_r \sim \Tilde{\lambda}$ \label{line:sampling}\\

Let $\calL = \{\}$.\\
 \For{$M = \emptyset, M_1,...,M_r$}{
    $\calL  = \calL \cup \text{BranchingProcess}(f, \InfluentialOracles, t+1, \InfluentialOracles' \cup \{g_i: i\in M\}, \eps, \delta)$\\}
 \Return{$\calL$}
  
\caption{Branching Process}
\label{alg:phase-one-branch}
\end{algorithm}

\begin{lemma}
\label{lemma:phase-one-correctness}
With probability at least $1-\delta$, at least one of the subsets \Cref{alg:phase-one-branch} returns is a set of coordinate oracles  to  $B\subseteq T$  such that
\begin{equation}\label{eq:B T condition}
\E_z\Bigg[\sum_{\substack{S \subseteq \specialset\setminus B \\ |S| > \fourierCutoff}} \hat{f_{B \to z}}(S)^2\Bigg] \leq \epsilon^2/4.	
\end{equation}
\end{lemma}
The reason for \Cref{eq:B T condition} becomes clear in \Cref{subsection:phase-two-lower-levels}, where we show that assuming the inequality, we lose at most an additive error of $\eps/2$ to the nearest $k$-junta if we ignore the Fourier mass above level $\fourierCutoff$ after restricting $B$. 
As before, in order to prove the above lemma, we prove a claim capturing the algorithm's progress towards satisfying \Cref{eq:B T condition}.

We denote the event $\calE$ that in the entire execution of \Cref{alg:phase-one-branch} all of the $\tilde{\lambda}_U$ were $\eps^2/48 \cdot \binom{|\InfluentialOracles|}{\fourierCutoff}$ close to the real $\lambda_U$. We note that by a union bound,  this happens with probability at least $1-\delta/2$. 

Suppose again that $\specialset$ is the (unknown) set of $k$ coordinates for which the best $k$-junta approximating $f$ is a junta on $\specialset$. 
If $\specialset$ has Fourier mass less than $\eps^2/4$ above level $\fourierCutoff$ then one of the subsets that \Cref{alg:phase-one-branch} will return is the empty set, which satisfies the claim. Therefore, henceforth we assume that $\specialset$ has at least $\eps^2/4$ Fourier mass above level $\fourierCutoff$. 
We show that in such a case, each $M_i$ for $i=1, \ldots, r$ will be a subset of $T$ with probability at least $\Omega(\eps^2)$.

\begin{claim}
\label{claim:prob-of-sample}
Assume $\InfluentialOracles'$ are coordinate oracles to $\InfluentialCoords' \subseteq \specialset$.
Suppose also that 
$$\E_z\Bigg[\sum_{\substack{S \subseteq \specialset\setminus \InfluentialCoords' \\ |S| > \fourierCutoff}} \hat{f_{\InfluentialCoords' \to z}}(S)^2\Bigg]  > \epsilon^2/4.$$
Then, conditioned on $\calE$, when running the Branching Process on $\InfluentialOracles'$, each $M_i$ will be with probability at least $\eps^2/40$  a collection of $\fourierCutoff$ new coordinate oracles to coordinates in  $\specialset$.%
\end{claim}

\begin{proof}
Similar to the proof of \Cref{claim:Sprime-correctness}, denote by $f_{z} = (f_{\avg, \InfluentialCoords})_{\InfluentialCoords' \to z}$, and note that $f'$ is up to relabeling of coordinates the same function as $f_z$. Denote $V\subseteq\specialset$ as the part of the target set we have not yet sampled, so $V = \specialset\setminus \InfluentialCoords'$. Then, using our assumption, we have that 
\begin{align*}
    \eps^2/4 &<  \E_z\Bigg[\sum_{\substack{S \subseteq V \\ |S| > \fourierCutoff}} \hat{f_{\InfluentialCoords' \to z}}(S)^2\Bigg] \\
    &=\sum_{\substack{S \subseteq V \\ |S| > \fourierCutoff}}\;\;\sum_{\substack{R\subseteq [n] :\\R\cap\overline{\InfluentialCoords'} = S }}\hat{f}(R)^2 \tag{\Cref{fact:expected_square}}\\
    &= \sum_{\substack{S \subseteq V \\ |S| > \fourierCutoff}}\;\;\sum_{\substack{R\subseteq \InfluentialCoords :\\R\cap\overline{\InfluentialCoords'} = S }}\hat{f}(R)^2
    \tag{if $R\not\subseteq\InfluentialCoords$ then $R\cap\overline{\InfluentialCoords'} \neq S$ }
    \\
    &= \sum_{\substack{S \subseteq V \\ |S| > \fourierCutoff}}\;\;\sum_{\substack{R\subseteq \InfluentialCoords :\\R\cap\overline{\InfluentialCoords'} = S }}\hat{f_{\avg, \InfluentialCoords}}(R)^2
    \\
    &= \E_z\Bigg[\sum_{\substack{S \subseteq V \\ |S| > \fourierCutoff}} \hat{f_{z}}(S)^2\Bigg].
\end{align*}

Next, by applying \Cref{theorem:ninf-ub-lb-set}, we have that
\begin{align*}\sum_{\substack{U \subseteq V\\|U| = \fourierCutoff}}\lambda_U[f_z] &\geq \frac{1}{2}\sum_{U \subseteq V} \NormInf_U^{\leq k}[f_z] 
 \ge \frac{1}{2} \sum_{U \subseteq V:|U|=\fourierCutoff} \;\sum_{S: U \subseteq S \subseteq V} \frac{\hat{f_z}(S)^2}{\binom{|S|}{|U|}} = \frac{1}{2}\sum_{\substack{|S|>\fourierCutoff\\ S\subseteq V}}\hat{f_z}(S)^2
\end{align*}
Then, using the assumption that $\calE$ happens, the $\tilde{\lambda_U}$ are $\frac{\eps^2}{48\cdot\binom{|\InfluentialCoords|}{\fourierCutoff}}$-accurate, and we get that
$$\sum_{\substack{U \subseteq V\\|U| = \fourierCutoff}}\Tilde{\lambda_U}[f_z] \geq \frac{1}{2}\sum_{\substack{|S|>\fourierCutoff\\ S\subseteq V}}\hat{f_z}(S)^2 - \frac{\eps^2}{48\cdot\binom{|\InfluentialCoords|}{\fourierCutoff}}\cdot \binom{k}{\fourierCutoff} \geq \frac{1}{2}\sum_{\substack{|S|>\fourierCutoff\\ S\subseteq V}}\hat{f_z}(S)^2 - \frac{\eps^2}{48}.$$
On the other hand, again by applying \Cref{theorem:ninf-ub-lb-set}, we have that
$$\sum_{\substack{U \subseteq\InfluentialCoords \\|U| = \fourierCutoff}}\lambda_U[f_z] \leq 3\sum_{\substack{U \subseteq\InfluentialCoords \\|U| = \fourierCutoff}}\NormInf_U[f_z] \leq 3\W^{\geq \fourierCutoff}[f_z] \leq 3. $$
This implies that $\sum_U \tilde{\lambda_U} \leq 3 + \frac{\eps^2}{48\cdot\binom{|\InfluentialCoords|}{\fourierCutoff}}\cdot \binom{|\InfluentialCoords|}{\fourierCutoff} \leq 4$. Overall, the probability to sample $U\subseteq V$ is at least 
$$\frac{1}{4}\left(\frac{1}{2}\sum_{\substack{|S|>\fourierCutoff\\ S\subseteq V}}\hat{f_z}(S)^2 - \frac{\eps^2}{48} \right ) = \frac{1}{8} \sum_{\substack{|S|>\fourierCutoff\\ S\subseteq V}}\hat{f_z}(S)^2 - \frac{\eps^2}{4\cdot48}.$$
Taking an expectation over $z$, we see that the probability to sample a subset of $V$ is at least 
\[\E_z \Bigg [\frac{1}{8} \sum_{\substack{|S|>\fourierCutoff\\ S\subseteq V}}\hat{f_z}(S)^2 - \frac{\eps^2}{4\cdot48}\Bigg ]\geq \frac{1}{8} \cdot \frac{\eps^2}{4} - \frac{\eps^2}{4\cdot 48} \geq \frac{\eps^2}{40}. \qedhere\]
\end{proof}

We are now ready to prove \Cref{lemma:phase-one-correctness}.
\begin{proof}[Proof of \Cref{lemma:phase-one-correctness}]
 By \Cref{claim:prob-of-sample}, if our special set $\specialset$ has at least $\epsilon^2/4$ mass on the levels above $\fourierCutoff$, then if we sample according to our distribution $\tilde{\lambda} = \{\tilde{\lambda}_U\}_{|U|=\fourierCutoff}$, we will see $U \subseteq \specialset$ with probability at least $\eps^2/40$. Then, if we sample $r = O(\eps^{-2})$ subsets in \Cref{alg:phase-one-branch}, applying the multiplicative Chernoff bound in \Cref{fact:chernoff}, we see at least one subset of $\specialset$ with probability at least $p \geq 0.9$ each time we  sample $M_1,...,M_r$ in \Cref{alg:phase-one-branch}. In order for \Cref{alg:phase-one-branch} to successfully find $B_i$ with the desired  property, it suffices to have sampled from $T$ at least $\branchProcessDepth$ times in our  branching process. Therefore, we can treat our $N:=(3\branchProcessDepth+\log(2/\delta))$  depth branching process as a $X = \Binomial(N, p)$ random variable.  Applying a  standard Chernoff bound (second case in~\Cref{fact:chernoff}), we have that our probability of failure is 
 \begin{align*}
 \Pr[X < \branchProcessDepth] = & \Pr[\overline{X} < \tfrac{\branchProcessDepth}{N}] \\
 =& \Pr[\overline{X} < 0.9 - (0.9- \tfrac{\branchProcessDepth}{N})] \\
 \leq & \exp (-2N(0.9-\tfrac{\branchProcessDepth}{N})^2) \tag{Using \Cref{fact:chernoff}}\\
\leq &\exp (-2N(0.81-2\tfrac{\branchProcessDepth}{N} )) \\
 \leq & \exp(-1.5N + 4 \branchProcessDepth) \\
 \le & \exp(-\log(2/\delta)) = \delta/2.
 \end{align*}
  This shows that, by a union bound with event $\calE$, one of the branches of our algorithm find's a $B_i$ satisfying \Cref{eq:B T condition} with probability at least $1-\delta$.
\end{proof}

\begin{claim}
\label{claim:phase-one-query-complexity}
The query complexity of phase one of the algorithm for constant $\delta$ (failure probability) is $2^{\tilde{O}(\sqrt{k/\eps})}$.
\end{claim}

\begin{proof}
All of our queries to $f$ in phase one come from estimating fourier coefficients using \Cref{claim:estimate Fourier non-Boolean} in \Cref{alg:lambda_estimate_set}. 
We require that the estimated Fourier coefficients be accurate to within $1/\poly(k,1/\eps)$ with confidence $1- O(1/\ell) = 1- 2^{-\tilde{\Omega}(\sqrt{k/\eps})}$, which is possible via \Cref{fact:chernoff} with query complexity $\poly(k/\eps)$. However, we do this $O(\ell) = 2^{\tilde{O}(\sqrt{k/\eps}})$ times during the branching process, which yields the final overall query complexity. 
\end{proof}

\subsection{Phase Two: The Lower Levels}
\label{subsection:phase-two-lower-levels}

Now, we are ready to use \Cref{alg:phase-one-branch}. Our strategy will be to take the subsets outputted from \Cref{alg:phase-one-branch} one at time, randomly fixing those coordinates, and then treating this restricted version of $f$ as if all its Fourier mass were below level $\fourierCutoff$ (recall that $\fourierCutoff = \sqrt{\eps k}$).
Let $T$ be the target set of size $k$ on which there exists a $k$-junta which best approximates $f$.
Assume that the first part of the algorithm is successful in yielding at least one $B\subseteq T$ such that:
\begin{equation}\label{eq:assumption on T}\E_{z \in \pmone^B}\bigbracket{\sum_{\substack{S \subseteq \specialset \setminus B \\ |S| > \fourierCutoff}} \hat{f_{B \to z}}(S)^2} \leq \epsilon^2/4.\end{equation}
Let $g$ be the maximizer of $\max_{g' \in \calJ_T} \E[fg']$. 
Recall that by \Cref{claim:maximum-correlation-junta} we have that
$g = \sgn(f_{\avg,T})$ and 
\begin{align}
    \corr(f, \calJ_T) = \E[fg] = \E_{y \in \BooleanHypercube{T}} [|f_{\avg, T}(y)|]
     &=\E_{z \in \BooleanHypercube{B}} \E_{x \in \BooleanHypercube{T \setminus B}}\bigg | (f_{\avg, T})_{B \to z}(x)\bigg|\\
     &=\E_{z \in \BooleanHypercube{B}} \E_{x \in \BooleanHypercube{T \setminus B}}
     \bigg | \sum_{S \subseteq T\setminus B} \hat{f_{B \to z}}(S)\chi_S(x)\bigg|\label{eq:latter}
     \end{align}
     Furthermore, using the assumption in Eq.~\eqref{eq:assumption on T} it is an easy calculation to show that \eqref{eq:latter} equals
     $$
     \E_{z \in \BooleanHypercube{B}} \E_{x \in \BooleanHypercube{T \setminus B}}
     \bigg | \sum_{S \subseteq T\setminus B, |S|\le \fourierCutoff} \hat{f_{B \to z}}(S)\chi_S(x)\bigg| \pm \eps/2.
     $$

Similarly, for any  set $U\subseteq \InfluentialCoords$ of size $k$ containing $B$ (think of $U$ as a candidate for $T$) we have that the best correlation between a junta-on-$U$ and $f$ is 
\begin{equation}
\label{eq:best junta on U}
\corr(f, \calJ_{U}) = \E_{z \in \BooleanHypercube{B}} \E_{x \in \BooleanHypercube{U \setminus B}}
     \bigg | \sum_{S \subseteq U\setminus B} \hat{f_{B \to z}}(S)\chi_S(x)\bigg|.
\end{equation}
Now, however, the right hand side in Eq.~\eqref{eq:best junta on U} is not necessarily approximated by the low-degree counterpart as above for $T$.
Indeed, we would like to estimate Eq.~\eqref{eq:best junta on U} for all candidates $U\subseteq \InfluentialCoords$ of size $k$ containing $B$, and pick the set with best estimated correlation.
Based on our assumption on $T$, we can replace $\sum_{S \subseteq U\setminus B} \hat{f_{B \to z}}(S)\chi_S(x)$ with its low-degree part
$\sum_{S \subseteq U\setminus B, |S|\le \fourierCutoff} \hat{f_{B \to z}}(S)\chi_S(x)$ for $U=T$, but its not clear whether we can do it in general.

In particular, if $U$ satisfies 
\begin{equation}\label{eq:U conc}
\E_{z\in \pmone^{B}}\bigbracket{\sum_{\substack{S \subseteq U \setminus B, \\ |S| > \fourierCutoff}} \;\hat{f_{B \to z}}(S)^2} >\epsilon^2/4,
\end{equation}
then taking the low-degree part can give an overestimate to the correlation with the best junta on $U$.\footnote{To see a simple example of how this can happen, consider $f(x,y) = 1 -x- y+xy$. Then one can verify that $\E[|f(x,y)|] = 1 < 1.5 = \E[|1-x-y|]$.}
We settle for an estimate that is $\eps$-accurate for the target set $T$ assuming it satisfies \Cref{eq:assumption on T}, and is not overestimating by more than $\eps$ for any other set $U\supseteq B$ of size $k$.
Towards this goal, we first apply a noise operator that would essentially eliminate most of the contribution from sets larger than $\sqrt{k/\eps} \log(1/\eps)$ regardless of whether $U$ satisfies Eq.~\eqref{eq:U conc} or not. 
This is captured by the following claim.

\begin{claim}\label{claim:cor noise to cor noise low deg}
Let $\rho = 1-\sqrt{\eps/k}$, $z\in \pmone^B$ and denote by $h = f_{B\to z}$ and $h^{\low} = h^{\leq (\sqrt{k/\eps})\cdot \log(1/\eps)}$ (i.e., $h^{\low}$ is the truncated Fourier expansion of $h$ that zeroes out all Fourier coefficients above level $(\sqrt{k/\eps})\cdot \log(1/\eps)$).
For any $U:B \subseteq U \subseteq \InfluentialCoords$ it holds that
$$
\bigg|\corr\left(T_\rho h, \calJ_U\right) - \corr\left(T_\rho h^{\low}, \calJ_U\right)\bigg| \leq \eps.
$$

\end{claim}
 
\begin{proof}
We have
\begin{align*}
& \bigg|\corr\left(T_\rho h, \calJ_U\right) - \corr\left(T_\rho h^{\low}, \calJ_U\right)\bigg|\\
   &= \Bigg|\E_{\substack{x \in \pmone^{U \setminus B}}}
     \Big | \sum_{S \subseteq U\setminus B} \hat{h}(S)\chi_S(x)\rho^{S}\Big|
     \;-\;\E_{\substack{x \in \pmone^{U \setminus B}}}
     \Big |\sum_{\substack{S \subseteq U\setminus B,\\|S|\le (\sqrt{k/\eps}) \cdot \log(1/\eps)}} \hat{h}(S)\chi_S(x) \rho^{S}\Big|\Bigg|     \\
      &\le 
\E_{\substack{x \in \pmone^{U \setminus B}}}
     \Bigg |\sum_{\substack{S \subseteq U\setminus B,\\ |S|> (\sqrt{k/\eps}) \cdot \log(1/\eps)}} \hat{h}(S)\chi_S(x)\rho^{|S|} \Bigg|
     \\
     &\le \sqrt{\E_{\substack{x \in \pmone^{U \setminus B}}}
    \Bigg(\sum_{\substack{S \subseteq U\setminus B,\\|S|>(\sqrt{k/\eps}) \cdot \log(1/\eps)}} \hat{h}(S)\chi_S(x) \rho^{|S|}\Bigg)^2} \\
    &= \sqrt{\sum_{\substack{S\subseteq U\setminus B, \\ |S| > (\sqrt{k/\eps}) \cdot \log(1/\eps)}}\hat{h}(S)^2\rho^{2|S|}} 
    \le \sqrt{\rho^{2(\sqrt{k/\eps}) \cdot \log(1/\eps)}} 
    \le \eps.\qedhere
 \end{align*}
\end{proof}

Next, we show that applying a noise operator to $f$ does not affect its correlation with a set $U$ of size $k$, under the condition that most of the Fourier mass of $f_{B\to z}$ falls on the lower levels, i.e., $\E_z\bigbracket{\sum_{\substack{S \subseteq U \setminus B,  |S| \geq \sqrt{k}}} \;\hat{f_{B \to z}}(S)^2} \leq \epsilon^2/4.$ Recall that this is what was guaranteed with high probability from the output of \Cref{alg:phase-one-branch} for our target set $\specialset$. 
\begin{claim}\label{claim:cor to cor noise}
Let $\rho = 1-\sqrt{k/\eps}$. 
 Given $U: B \subseteq U \subseteq \InfluentialCoords$ such that
 $\E_z\bigbracket{\sum_{\substack{S \subseteq U \setminus B,\\  |S| \geq \fourierCutoff}} \;\hat{f_{B \to z}}(S)^2} \leq \epsilon^2/4,$
we have that 
$$
\left |\E_z \corr(T_\rho (f_{B \to z}), \calJ_U) - \E_z \corr(f_{B \to z}, \calJ_U)\right| \leq 1.2 \eps.
$$
\end{claim}
 
\begin{proof}
Similar to the proof of \Cref{claim:cor noise to cor noise low deg}, we have
\begin{align*}
   & \Bigg|\E_{\substack{z \in \pmone^{B}\\x \in \pmone^{U \setminus B}}}
     \Big | \sum_{S \subseteq U\setminus B} \hat{f_{B \to z}}(S)\chi_S(x)\Big|
     \;-\;\E_{\substack{z \in \pmone^{B}\\x \in \pmone^{U \setminus B}}}
     \Big |\sum_{\substack{S \subseteq U\setminus B}} \hat{f_{B \to z}}(S)\chi_S(x) \cdot \rho^{|S|}\Big|\Bigg|     \\
      &\le 
\E_{\substack{z \in \pmone^{B}\\x \in \pmone^{U \setminus B}}}
     \Bigg |\sum_{\substack{S \subseteq U\setminus B}} \hat{f_{B \to z}}(S)\chi_S(x) (1-\rho^{|S|}) \Bigg|\\
     &\le \sqrt{\E_{\substack{z \in \pmone^{B}\\x \in \pmone^{U \setminus B}}}
    \left(\sum_{\substack{S \subseteq U\setminus B}} \hat{f_{B \to z}}(S)\chi_S(x) (1-\rho^{|S|})\right)^2} 
    \\&= \sqrt{\E_{z \in \pmone^{B}}\Bigg[\sum_{\substack{S\subseteq U\setminus B}}\hat{f_{B\to z}}(S)^2\cdot (1-\rho^{|S|})^2\Bigg]} \\
    &\leq 
    \sqrt{\E_{z \in \pmone^{B}}\Bigg[\sum_{\substack{S\subseteq U\setminus B:|S|\le \fourierCutoff}}\hat{f_{B\to z}}(S)^2\cdot (1-\rho^{|S|})^2 + \sum_{\substack{S\subseteq U\setminus B:|S|> \fourierCutoff}}\hat{f_{B\to z}}(S)^2\cdot (1-\rho^{|S|})^2\Bigg]}\\
    &\le\sqrt{ (1-\rho^{\fourierCutoff})^2 + \eps^2/4} 
    \le \sqrt{ \eps^2 + \eps^2/4} \le 1.2 \cdot \eps.\qedhere
     \end{align*}
\end{proof}

The next lemma gives an algorithm that on any $B$,  satisfying \Cref{eq:assumption on T},  outputs $U: B \subseteq U \subseteq \InfluentialCoords$
with $\corr(f, \calJ_U) \ge \corr(f, \calJ_T)-O(\eps)$, with high probability. 
\begin{lemma}[Algorithm and Analysis for Phase-Two]\label{lemma:phase-two-correctness}
Let $\eps,\delta>0$. There's an algorithm that with probability at least $1-\delta$, gives $\eps$-accurate estimates $\tilde{c_U}$ to 
$$
{c_U} =  \E_{z\in \pmone^B}\E_{x \in \BooleanHypercube{T\setminus B}}\bigg | \sum_{S\subseteq U\setminus B:|S|\le \sqrt{k/\eps} \cdot \log(1/\eps)}\hat{f_{B\to z}}(S) \chi_S(x)\rho^{|S|}\bigg |
$$
for all $U: B \subseteq U \subseteq \InfluentialCoords$ of size $k$ simultaneously.
We  return  $(U,\tilde{c_U})$ for the set $U$ with maximal $\tilde{c_U}$.

\begin{description}
\item[Complexity] The procedure uses $\log(1/\delta)2^{\tilde{O}(\sqrt{k/\eps})}$ queries and runs in time $\log(1/\delta)2^{k\cdot \tilde{O}(1/\eps)}$.
\item[Correctness]
In the case where all estimates are $\eps$-accurate, the following holds.
If $B\subseteq T$  satisfies \Cref{eq:assumption on T}, the above procedure would return $(U, \tilde{c_U})$ with $\tilde{c_U} \ge \corr(f, \calJ_T) - 3.2\eps$.
Moreover, regardless of whether $T$ and $B$ satisfy \Cref{eq:assumption on T}, we have $\tilde{c_U} \le \corr(f, \calJ_U) + 2\eps$.
\end{description}

\end{lemma}
\begin{proof}
First we show that we can estimate all ${c_U}$ up to error $\eps$ simultaneously with high probability using the aforementioned query complexity and running time.
We sample $t = O(\log(1/\delta)/\eps^2)$ different $z\in \pmone^B$, and estimate for each value of $z$ the Fourier coefficients of $\hat{f_{B \to z}}(S)$ of all sets $S \subseteq \InfluentialCoords$ of size at most $\zeta = \sqrt{k/\eps} \cdot \log(\tfrac{2}{\eps})$ up to additive error $\eps/\binom{k}{\le \zeta} = 2^{-\tilde{\Omega}(\sqrt{k/\eps})}$ with probability $1- \frac{\delta}{t\cdot \binom{k}{\leq \zeta}}$, which is possible via \Cref{fact:chernoff} with $\log(1/\delta)2^{\tilde{O}(\sqrt{k/\eps})}$ queries. 
\Cref{fact:chernoff} guarantees that with probability $1-\delta$ for all sampled $z$, all estimated low-degree Fourier coefficients are within the additive error bound, in which case we have estimates for all $c_U$ up to error $\eps$ simultaneously with probability $1-\delta$.

Next, we show the correctness of the procedure.
On the one hand, in the assumed case, i.e., that $T$ satisfies $\E_z\bigbracket{\sum_{\substack{S \subseteq T \setminus B,  |S| \geq \fourierCutoff}} \;\hat{f_{B \to z}}(S)^2} \leq \frac{\epsilon^2}{4},$
we will have by  \Cref{claim:cor noise to cor noise low deg} and \Cref{claim:cor to cor noise} that \begin{equation}\label{eq:ct:lb}c_T \ge  \corr(f, \calJ_T) - 2.2 \eps\end{equation}

    Since we output the set $U$ with maximal $\tilde{c_U}$, and since all estimates are correct up to $\eps$ 
    we know that we output $U$ with \begin{equation}\label{eq:cu:ct}\tilde{c_U} \ge \tilde{c_T} \ge c_T-\eps.\end{equation}
     
  Combining \Cref{eq:ct:lb,eq:cu:ct} together we get $$\tilde{c_U}  \ge c_T - \eps \ge  \corr(f, \calJ_T) - 3.2\eps.$$
  
  We move to prove the furthermore part, i.e., that $\tilde{c_U} \le \corr(f,\calJ_U) + 2\eps$ regardless of whether $T$ and $B$ satisfy \Cref{eq:assumption on T}.
  We start by showing that for any set $U$ (whatsoever) we have that $\corr(f,\calJ_U) \ge c_U - \eps$.
Indeed, by \Cref{claim:cor noise to cor noise low deg} we have $${c_U} \approx_{\eps}\E_{\substack{z \in \pmone^{B}\\x \in \pmone^{U \setminus B}}}
     \Big | \sum_{S \subseteq U\setminus B} \hat{f_{B \to z}}(S)\chi_S(x)\rho^{|S|}\Big|$$ and since the noise operator can only reduce $\ell_1$-norm (see \Cref{fact:noise-op-l1-norm}), we see that for all $z\in \pmone^{B}$ it holds that
     $$
     \E_{x \in \pmone^{U \setminus B}}
     \Big | \sum_{S \subseteq U\setminus B} \hat{f_{B \to z}}(S)\chi_S(x)\rho^{|S|}\Big| \le 
     \E_{x \in \pmone^{U \setminus B}}
     \Big | \sum_{S \subseteq U\setminus B} \hat{f_{B \to z}}(S)\chi_S(x)\Big|
     $$
     
Thus,
     \begin{align*}
     c_U  &\le \eps +   
     \E_{\substack{z \in \pmone^{B}\\x \in \pmone^{U \setminus B}}}
     \Big | \sum_{S \subseteq U\setminus B} \hat{f_{B \to z}}(S)\chi_S(x)\rho^{|S|}\Big| \nonumber \\
     &\le  \eps +
     \E_{\substack{z \in \pmone^{B}\\x \in \pmone^{U \setminus B}}}
     \Big | \sum_{S \subseteq U\setminus B} \hat{f_{B \to z}}(S)\chi_S(x)\Big|
      = \eps + 
     \corr(f, \calJ_U)\end{align*}
  Since $|c_U - \tilde{c_U}| \le \eps$, we get that  $\tilde{c_U} \le c_U + \eps \le \corr(f, \calJ_U) + 2\eps$.
\end{proof}

After phase one, we can apply \Cref{lemma:phase-two-correctness} to each $B$ from phase one, and get a set $U_B: B \subseteq U_B \subseteq \InfluentialCoords$ of size $k$, along with an estimate of the correlation of $f$ to $\calJ_{U_B}$. This leads to the proof of \Cref{theorem:main-result} which we restate next.
\begin{theorem}
Given a Boolean function $f: \BooleanHypercube{n} \to \BooleanHypercube{}$, it is possible to estimate the distance of $f$ from the class of $k$-juntas to within additive error $\epsilon$ with probability $2/3$ using 
$ 2^{\tilde{O}(\sqrt{k/\eps})}$ adaptive queries to $f$. In particular, when $\epsilon$ is constant, this yields a $2^{\tilde{O}(\sqrt{k})}$-query algorithm. However, the algorithm still requires $\exp(k/\eps)$ time.
\end{theorem}

\begin{proof}
Let $\eps_0 = \eps/6$
\begin{enumerate}
    \item We first apply the result of \cite{junta-coordinate-oracles} to reduce the  down to only $\poly(k,1/\eps_0)$ coordinates.
This incurs a loss in correlation of at most $\eps_0$, and fails with probability at most $\delta_1$, which we can set to be $1/20$, by \Cref{corollary:get-oracles}.
\item Next, we apply our \Cref{theorem:improved-dmn}, which reduces the number of oracles we have to consider down to $O(k/\eps_0^2)$, incurs an additive loss in correlation of at most $\eps_0$, and fails with probability at most $\delta_2 = 1/20$.
\item Then, we run phase 1 of our algorithm, which fails with probability at most $\delta_3 = 1/20$ by \Cref{lemma:phase-one-correctness}.
\item 
Finally, we apply \Cref{lemma:phase-two-correctness} to every $B$ outputted by \Cref{alg:phase-one-branch} to get a set $U_B$ and an estimate $\tilde{C_{U_B}}$ for the correlation of $f$ with $\calJ_{U_B}$ 
We iterate on all sets $B$ returned by phase-1 and return $U_B$ with the highest estimate of correlation.

    There are $\ell = O(\frac{1}{\eps_0^2})^{3\sqrt{k/\eps_0} + \log(2/\delta_3)} = 2^{\tilde{O}(\sqrt{k/\eps_0})}$ branches, and thus if we apply the algorithm from \cref{lemma:phase-two-correctness} with $\delta = 1/(20\ell)$, we get that all this step fail with probability at most $1/20$ by a union bound. 
\end{enumerate}

By a union bound, each of these steps succeeds with probability at least $1 - 4/20 \geq 2/3$.
In the case all steps succeeds, we return a set $U$ with $\tilde{c_U} \ge \corr(f, \calJ_{n,k})-5.2\eps_0$. In addition, the moreover part in \Cref{lemma:phase-two-correctness} guarantees that $\tilde{c_U} \le \corr(f, \calJ_{U}) + 2\eps_0 \le \corr(f, \calJ_{n,k}) + 2\eps_0$. We get that the returned value is within $5.2\eps_0 < \eps$ of $\corr(f,\calJ_{n,k})$.
Finally, since $\dist(f,\calJ_{n,k}) = \frac{1+\corr(f,\calJ_{n,k})}{2}$ we get that $\frac{1+ \tilde{c_U}}{2}$ is an $\eps/2$-accurate approximation of $\dist(f,\calJ_{n,k})$. Finally, we note that the query complexities of phase 1 and phase 2 are both $2^{\Tilde{O}(\sqrt{k/\eps})}$, but the runtime is exponential due to \cref{lemma:phase-two-correctness}.
\end{proof}

Finally, we mention that if our goal is not to estimate to correlation with the nearest $k$-junta to $f$, but rather to simply estimate the most amount of Fourier mass any subset of $k$ variables contains, then we have the following theorem with an improved dependence on $\eps$:
\begin{theorem}
\label{theorem:main-result-mass}
Given a Boolean function $f: \BooleanHypercube{n} \to \BooleanHypercube{}$, it is possible to estimate the most mass any subset of at most $k$ variables of $f$ has to within additive error $\eps$ with probability $2/3$ using 
$ 2^{\tilde{O}(\sqrt{k}\log(1/\eps))}$ adaptive queries to $f$. In particular, when $\epsilon$ is constant, this yields a $2^{\tilde{O}(\sqrt{k})}$-query algorithm. However, the algorithm still requires $\exp(k\log(1/\eps))$ time.
\end{theorem}

We leave the proof of this theorem, which involves simple modifications to the algorithm presented in this section, to \Cref{appendix:k-mass-estimation}.

\subsection{Proof of \Cref{theorem:ninf-ub-lb-set}}
\label{subsection:proof-ninf-ub-lb-set}
We now present the proof of \Cref{theorem:ninf-ub-lb-set}.
\begin{proof}[Proof of \Cref{theorem:ninf-ub-lb-set}]
The proof is very similar to the previous proof of \Cref{theorem:ninf-ub-lb}, so we explain how to modify it to this case.	

We express $\lambda_U$ in terms of the Fourier spectrum of $f$.
\begin{align*}
\lambda_U &= \sum_{m=0}^{2|U|\log(10k)} \sum_{S: S \supseteq U} \hat{f}(S)^2 \cdot \Pr_{J \subseteq_{p^{m}} [\ell]}[ S \cap J = U]
\\ &= \sum_{m=0}^{2|U|\log(10k)} \sum_{S: S \supseteq U} \hat{f}(S)^2 \cdot \Pr_{J \subseteq_{p^{m}} [\ell]}[ |S \cap J| = |U|] \cdot \frac{1}{\binom{|S|}{|U|}}
\\&=\sum_{S: S \supseteq U} \frac{\hat{f}(S)^2}{\binom{|S|}{|U|}} \cdot \sum_{m=0}^{2|U|\log(10k)} \Pr_{J \subseteq_{p^{m}} [\ell]}[ |S \cap J| = |U|]\end{align*}
It suffices to show that for any non-empty set $S$ of size at least $|U|$ and at most $k$ it holds that
\begin{equation}\label{eq:2}
\sum_{m=0}^{2|U|\log(10k)} \Pr_{J \subseteq_{p^{m}} [\ell]}[ |S \cap J| = |U|] \in [1/2,3]\;.
\end{equation}
Again, we can analyze the sum on the left hand side of \Cref{eq:2} as the expected final value of $X$ in the following random process:

\begin{algorithm}
$X \leftarrow 0$ \\
\For{$i=1, 2, \ldots, 2|U|\log(10k)$} {
    \If{$|S| < |U|$} { 
        halt!
    }
    \If{$|S|=|U|$}{
        increase $X$
    }
	Sample $J_i \subseteq_p [\ell]$ \\
	$S \leftarrow S \cap J_i$
}
\end{algorithm}
By symmetry the expected value depends only on the size of the initial set $S$.
As before, we denote by $F_{t}$ its expected value starting with a set $S$ of size $t$ with an infinite horizon, and $F_t^{(i)}$ as the expected value of $X$ at the end of the above process with finite horizon $i$.
We start by analyzing $F_{|U|}$.
In this case, $X$ is a geometric random variable with stopping probability $1- p^{|U|}$.
Thus, its expectation is $$F_{|U|} = 1/(1-p^{|U|}) = 1/(1-(1-1/2|U|)^{|U|}) \in [2,3].$$
This implies that $F_{|U|}^{(2|U|\log(10k))} \leq F_{|U|} \leq 3$. For $t > |U|$ in the infinite horizon case we have the recurrence
\begin{equation}
	F_t = \sum_{a=0}^{t} F_a \cdot \Pr[\Binomial(t,p)=a]
	 = \sum_{a=|U|}^{t-1} F_a \cdot \Pr[\Binomial(t,p)=a]  + F_t \cdot \Pr[\Binomial(t,p)=t]
\end{equation}
or equivalently
\begin{equation}
	F_t\cdot \Pr[\Binomial(t,p)<t] = \sum_{a=|U|}^{t-1} F_a \cdot \Pr[\Binomial(t,p)=a]
\end{equation}

We prove by induction that for $t \ge |U|$ it holds that 
$F_t \le F_{|U|}$.
The claim clearly holds for $t=|U|$. For $t> |U|$ we can apply induction and get
$$
F_t\cdot \Pr[\Binomial(t,p)<t] \le  \sum_{a=|U|}^{t-1} F_{|U|} \cdot \Pr[\Binomial(t,p)=a]  \le F_{|U|}\cdot \Pr[\Binomial(t,p)<t],$$
and thus $F_t \le F_{|U|}$. This immediately implies that $F_t^{(2|U|\log(10k))} \leq F_t \leq 3$.
On the other hand we prove that $F_t^{(2|U|\log(10k))} \ge 1/2$ as long as $t\leq k$. To do so, we once again introduce the indicator random variable $Y_t^{(i)}$, where $t = |S|$, and which equals 1 if $|S|=|U|$ at some point during the above process before iteration $i$. We note that $Y_t^{(2|U|\log(10k))}$ is a lower bound for the value of $X$ in the above process, and $Y_t$ is a lower bound for the value of $X$ at the end of the infinite horizon process. We note that the case $|U|=1$ was already lower bounded in \Cref{subsection:proof-of-ninf-ub-lb}, where it was shown that $\E[Y_t^{(\log(10k))}] \ge 1/2$, and therefore $\E[Y_t^{(2|U|\log(10k))}] \ge 1/2$.
It remains to show that the $\E[Y_t^{(2|U|\log(10k))}] \ge 1/2$ is true for any set $|U| \ge 2$.

First, we show that  $\Pr[\Binomial(t,p)<|U|] \le \frac{1}{2} \Pr[\Binomial(t,p)=|U|]$.
Towards this goal, it would suffice to prove
that $3\le \Pr[\Binomial(t,p)=i+1]/\Pr[\Binomial(t,p)=i]$ for $i<|U|$ and $t\ge|U|+1$. 
This would suffice since in this case $$\sum_{i=0}^{|U|-1} \Pr[\Binomial(t,p)=i] \le \sum_{i=0}^{|U|-1} \frac{3^{i}}{3^{|U|}} \Pr[\Binomial(t,p)=|U|] \le \frac{1}{2} \cdot \Pr[\Binomial(t,p)=|U|].$$
Indeed, The ratio between the two aforementioned probabilities is 
$$
\frac{\Pr[\Binomial(t,p)=i+1]}{\Pr[\Binomial(t,p)=i]} = \frac{\binom{t}{i+1}}{\binom{t}{i}} \cdot  \frac{p^{i+1}(1-p)^{t-(i+1)}}{p^{i}(1-p)^{t-i}} = 
\frac{t-i}{i+1} \cdot \frac{p}{1-p} \ge \frac{2}{|U|}\cdot \frac{1-1/2|U|}{1/2|U|} =\frac{2-1/|U|}{1/2}\ge 3
$$
as needed.
Now, we claim that $\E[Y_t] = \Pr[Y_t = 1] \geq 2/3$ for all $t \geq 1$. The base case of $t=1$ is certainly true. Assuming we have $\Pr[\Binomial(t,p)<|U|] \le \tfrac{1}{2}\Pr[\Binomial(t,p)=|U|]$ we have 
\begin{align*}
	\E[Y_t]\cdot \Pr[ Bin(t,p)<t] &=  
    \sum_{a=|U|}^{t-1} \E[Y_a] \cdot Pr[\Binomial(t,p)=a]\\
    &\ge \Pr[\Binomial(t,p)=|U|] + 
    \sum_{a=|U|+1}^{t-1} \Pr[\Binomial(t,p)=a]\E[Y_a]\\
    &\geq \Pr[\Binomial(t,p) = |U|] + \frac{2}{3}\Pr[\Binomial(t,p) \in [|U|+1,t-1]] \\
    &= \frac{2}{3} \Pr[\Binomial(t,p) <t] - \frac{2}{3}\Pr[\Binomial(t,p)<|U|] + \frac{1}{3}\Pr[\Binomial(t,p)=|U|]\\
    &\ge \frac{2}{3}\Pr[\Binomial(t,p)<t] 
\end{align*}
which implies that $\E[Y_t] \geq 2/3$. Finally, let $A$ be the event that $S = \emptyset$ by iteration $2|U|\log(10k)$, and note that 
\begin{align*}
\Pr[A] &= \Pr[\Binomial(|S|, (1-\tfrac{1}{2|U|})^{2|U|\log(10k)} )= 0] \\
&\geq \Pr[\Binomial(k, e^{-\log (10k)}) = 0]  = \Pr[\Binomial(k, \tfrac{1}{10k}) = 0] \geq 0.9
\end{align*}
as was shown in the proof for \Cref{theorem:ninf-ub-lb} in \Cref{subsection:proof-of-ninf-ub-lb}. Finally, we claim that for all $t\geq 2$ we have that $\Pr[Y_t^{(2|U|\log 10k)}]\geq 1/2$. Indeed, we have that
\begin{align*}
    \Pr[Y_t^{(2|U|\log 10k))} = 1] &\geq \Pr[Y_t=1] -\Pr[\overline{A}] \ge \tfrac{2}{3}-0.1    \geq \tfrac{1}{2}.
\end{align*}
as desired, provided $|S|\leq k$.
\end{proof}

\section{Conclusions and Open Problems}
\label{section:conclusion}

We conclude by mentioning some future research directions. First, we believe some of the techniques discussed in this paper could lead to other interesting work in property testing, learning theory, or Boolean function analysis in general. 
In particular, the procedure in \Cref{alg:sampling} makes use of a random process to get access to an underlying junta, a subprocedure that could be useful in other learning or testing algorithms. 
In addition, we are able to approximate the quantities $\NormInf_i$ and $\NormInf_U$, that serve as key steps in our algorithms. These quantities seems natural on their own, and would likely find further applications in Analysis of Boolean functions. In particular, they seem to capture more accurately the intuition that ``influences measures the importance of coordinates''.
While the total influence of a Boolean function can be any number between $\var[f]$ and $n\cdot \var[f]$ the total normalized influence equals exactly $\var[f]$, and thus normalized influences can be seen as a distribution of the variance among the coordinates. 

Interestingly, our algorithms strongly resemble certain quantum algorithms. In particular, the sampling of coordinates is done through the Fourier distribution, a process which can be done much more efficiently with a quantum algorithm (querying $f$ in superposition, applying the Hadamard transform, and measuring). 
This idea was leveraged in \cite{quantum-junta-testing} to provide fast quantum algorithms for testing juntas in the standard property testing regime. 
Indeed, if the nearest $k$-junta to $f$ has its mass on higher levels (say above $\sqrt{k}$ or even $k/2$), then Fourier sampling is extremely effective and provides a cleaner way of sampling subsets according to the Fourier distribution than the related classical technique we provided in \Cref{section:main-result}. However, the issue arises when the nearest $k$-junta has Fourier mass on lower levels (below $\log k$ or even a constant, for example). In this case, it is not clear to us how quantum algorithms provide any advantage over classical ones. 
An open question is whether quantum Fourier sampling techniques can be applied in a more clever way to give faster algorithms in the tolerant testing paradigm.

Finally, a clear open question is how good of a lower bound one can prove on the query complexity of the tolerant junta testing problem. Our main result \Cref{theorem:main-result}, rules out strictly exponential-in-$k$ query lower bounds for $k$-junta distance approximation. \cite{junta-lbs-PRW20} proved a non-adaptive query complexity lower bound of $2^{k^\eta}$ for $(k, k, \eps_1,\eps_2)$-tolerant junta testing (given a particular choice of $0 < \eps_1 < \eps_2 < 1/2$), for any $0<\eta<1/2$. While this is quite close to our upper bound of $2^{\Tilde{O}(\sqrt{k})}$, our algorithm is highly adaptive, while the lower bound due to \cite{junta-lbs-PRW20} applies only to nonadaptive algorithms. 
Therefore, another interesting direction would be to explore whether any nontrivial lower bounds  apply to adaptive algorithms for tolerant (junta) testing and distance approximation.

\subsection*{Acknowledgements}
We thank Anindya De, Shafi Goldwasser, Amit Levi, and Orr Paradise for very helpful discussions.

\printbibliography

\appendix 

\section{Maximum $k$-Subset Fourier Mass Approximation}
\label{appendix:k-mass-estimation}
In this section, we sketch a proof of \Cref{theorem:main-result-mass}, which involves simple modifications and observations about our algorithm. The main difference is that we sample from the normalized influence subdistribution at a different Fourier level --  namely, we let $\fourierCutoff := \sqrt{k}$ and $\branchProcessDepth = k/\fourierCutoff = \sqrt{k}$ in \Cref{alg:lambda_estimate_set} and \Cref{alg:phase-one-branch}, respectively (recall that before, $\fourierCutoff = \sqrt{\eps k}$). This improves the query complexity dependence on $\eps$ in Phase 1.
\begin{claim}
\label{claim:phase-one-mass-query-complexity}
The query complexity of phase one of the algorithm for constant $\delta$ (failure probability) is $2^{\tilde{O}(\sqrt{k}\log(1/\eps))}$.
\end{claim}
\begin{proof}
The proof is analogous to the proof of \Cref{claim:phase-one-query-complexity}, so we just point out the differences. We still require our Fourier coefficients to be accurate to within $1/\poly(k,1/\eps)$, and we require confidence $1-O(1/\ell) =1 - 2^{\Tilde{\Omega}(\sqrt{k}\log(1/\eps))}$. However, now our branching process now has depth only $O(\sqrt{k})$, so we need only repeat this $O(\ell) = 2^{\tilde{O}(\sqrt{k}\log(1/\eps))}$ times, which yields the improved query complexity.
\end{proof}
In Phase 2, we argue that it is not necessary to apply a noise operator in order to only consider Fourier mass below level $\fourierCutoff$ after Phase 1. Recall that we applied this noise operator in \Cref{subsection:phase-two-lower-levels} in order to deal with the case that a particular $U$ satisfied
\begin{equation}\label{eq:U conc}
\E_{z\in \pmone^{B}}\bigbracket{\sum_{\substack{S \subseteq U \setminus B, \\ |S| > \fourierCutoff}} \;\hat{f_{B \to z}}(S)^2} >\epsilon^2/4.
\end{equation}
If this happened, then we could not rule out the possibility that taking the low-degree part of $f$ within $U$ gives an \textit{overestimate} to the correlation with the best $k$-junta. However, now we are not concerned with the junta correlation, but rather which set has the most mass, so we claim we do not have to worry about this possibility anymore. To see this, suppose we have identified $B \subseteq U$, and note that
\begin{align*}
    \sum_{S \subseteq U}\hat{f}(S)^2 &= \E_x[f(x)f_{\avg, U}(x)] \\
    &= \E_{z\in \BooleanHypercube{B}} \bigg [ \E_x[f_{B\to z}(x)(f_{\avg, U})_{B \to z}(x)]\bigg]\\
    &= \E_z \bigg [  \sum_{\substack{S \subseteq U \\ |S| \leq \fourierCutoff}}\hat{f_{B\to z}}(S)^2 + \sum_{\substack{S \subseteq U \\ |S| > \fourierCutoff}}\hat{f_{B\to z}}(S)^2\bigg ]\\
    &\geq \E_z \bigg [  \sum_{\substack{S \subseteq U \\ |S| \leq \fourierCutoff}}\hat{f_{B\to z}}(S)^2 \bigg ]. 
\end{align*}

Therefore, we no longer have to apply any noise operator, which negates the necessity of \Cref{claim:cor noise to cor noise low deg} and \Cref{claim:cor to cor noise}. It therefore suffices in \Cref{lemma:phase-two-correctness} to estimate the mass of each set, rather than the correlation, as 
$$m_U =\E_z \bigg [  \sum_{\substack{S \subseteq U \\ |S| \leq \fourierCutoff}}\hat{f_{B\to z}}(S)^2 \bigg ].$$
To do so, as in the proof of \Cref{lemma:phase-two-correctness} we let $t = O(\log(1/\delta)/\eps^2)$ be the number of random samples of $z$ we take. Then we estimate all the Fourier coefficients below level $\fourierCutoff$. This requires estimating $\hat{f}(S)$ for all $S \subseteq\InfluentialCoords$ of size at most $\fourierCutoff$ up to additive error $\eps/\binom{k}{\leq \kappa} = 2^{\tilde{\Omega}(\sqrt{k}\log(1/\eps))}$ with probability $1-\frac{\delta}{t\cdot \binom{k}{\leq \kappa}}$, which is possible via \Cref{fact:chernoff} with $\log(1/\delta)2^{\Tilde{O}(\sqrt{k}\log(1/\eps))}$ queries. The rest of our argument and algorithm is exactly the same as in \Cref{section:main-result}.

\end{document}